\documentclass[a4paper,11pt,english]{article}
\pdfoutput=1

\usepackage{ifdraft}
\ifdraft{\newcommand{\authnote}[3]{{\color{#3} {\bf  #1:} #2}}}{\newcommand{\authnote}[3]{}}

\newcommand{\gnote}[1]{\authnote{ Andr\'{a}s}{#1}{blue}}

\newif\ifcount
\ifdraft{\counttrue}

\usepackage[utf8]{inputenc}
\usepackage[T1]{fontenc}
\usepackage[left=1in,top=1in,right=1in,bottom=1in]{geometry} % better than fullpage
\usepackage[english]{babel}
\usepackage{verbatim}

\usepackage{tikz}
\usepackage{amssymb}
\usepackage{mathtools}
\usepackage{bbm}
\usepackage{upgreek}
\usepackage{mleftright}\mleftright
\usepackage{multirow}
\usepackage{amsthm} 
\theoremstyle{plain}
\usepackage{thmtools} 
\usepackage{thm-restate}

\usepackage{algorithm}
\usepackage{algcompatible}
\usepackage{float}
\usepackage{enumitem}
\makeatletter
\def\namedlabel#1#2{\begingroup
	#2%
	\def\@currentlabel{#2}%
	\phantomsection\label{#1}\endgroup
}

\usepackage{caption}
\usepackage{subcaption}

\usepackage{titling}
\usepackage[final=true,colorlinks = true,allcolors = {blue},]{hyperref}

\newcommand{\eps}{\varepsilon}

\newcommand{\ketbra}[2]{|#1\rangle\! \langle #2|}

\newcommand{\nrm}[1]{\left\lVert #1 \right\rVert}
\newcommand{\bigO}[1]{\mathcal{O}\left( #1 \right)}
\newcommand{\bigObig}[1]{\mathcal{O}\big( #1 \big)}
\newcommand{\bigOt}[1]{\widetilde{\mathcal{O}}\left( #1 \right)}

\newcommand{\diag}[1]{\mathrm{diag}\left( #1 \right)}

\newcommand{\vertiii}[1]{{\left\vert\kern-0.25ex\left\vert\kern-0.25ex\left\vert #1 
		\right\vert\kern-0.25ex\right\vert\kern-0.25ex\right\vert}}

\newcommand{\C}{\mathbb{C}}
\newcommand{\N}{\mathbb{N}}

\newcommand{\PM}{\mathcal{P}}

\newcommand{\pvp}{\vec{p}{\kern 0.45mm}'}
\let\oldnabla\nabla
\renewcommand{\nabla}{\oldnabla\!}

\DeclarePairedDelimiter\bra{\langle}{\rvert}
\DeclarePairedDelimiter\ket{\lvert}{\rangle}
\DeclarePairedDelimiterX\braket[2]{\langle}{\rangle}{#1 \delimsize\vert #2}
\newcommand{\underflow}[2]{\underset{\kern-60mm \overbrace{#1} \kern-60mm}{#2}}

\def\Pr{\mathrm{Pr}}

\long\def\ignore#1{}

\newtheorem{theorem}{Theorem}
\newtheorem{corollary}[theorem]{Corollary}
\newtheorem{lemma}[theorem]{Lemma}

\newtheorem*{claim*}{Claim}

\newtheorem{conj}[theorem]{Conjecture}

\usepackage{tikz,pgfplots}

\usetikzlibrary{backgrounds,fit,decorations.pathreplacing,calc}

\title{$\kern-2mm$Quadratic speedup for finding marked vertices by quantum walks$\kern-2mm$}
\author{
	Andris Ambainis\thanks{Faculty of Computing, University of Latvia. Supported by Latvian Council of
Science grant lzp-2018/1-0173 and QuantERA project QuantAlgo 680-91-034.}
	\and
	András Gilyén\thanks{QuSoft, CWI and University of Amsterdam, the Netherlands. Supported by ERC Consolidator Grant QPROGRESS and partially supported by QuantERA project QuantAlgo 680-91-034.}%\textsuperscript{\kern1.4mm,\textasteriskcentered}
	\and
	Stacey Jeffery\thanks{QuSoft and CWI, the Netherlands. Supported by an NWO Veni Innovational Research Grant under project number 639.021.752 and an NWO WISE Grant.}
	\and
	Martins Kokainis\thanks{Faculty of Computing, University of Latvia. Supported by QuantERA project QuantAlgo 680-91-034.}
}
\date{\today\vspace{-5mm}}

\newcommand\lr[1]{\left( #1 \right)}
\newcommand\lrv[1]{\left|  #1 \right|}
\newcommand\lrb[1]{\left\lbrace #1 \right\rbrace}
\newcommand{\smallO}[1]{ {o}\left( #1 \right)}
\renewcommand{\check}{\mathtt{Check}}
\newcommand{\setup}{\mathtt{Setup}}
\newcommand{\update}{\mathtt{Update}}
\newcommand{\checkingcost}{\mathsf{C}}
\newcommand{\setupcost}{\mathsf{S}}
\newcommand{\updatecost}{\mathsf{U}}
\newcommand{\Reg}{\mathsf{R}}
\newcommand{\Hi}{\mathcal{H}}
\newcommand{\barO}{\bar{0}}
\newcommand{\psuccess} {p_{\textrm{success}}}
\theoremstyle{definition}
\newtheorem{exmp}{Example}[section]
\pgfplotsset{compat=1.13}
\begin{document}
	
\maketitle 
%TC:break Part: Abstract

\begin{abstract}
A quantum walk algorithm can detect the presence of a marked vertex on a graph quadratically faster than the corresponding random walk algorithm (Szegedy, FOCS 2004). However, quantum algorithms that actually find a marked element quadratically faster than a classical random walk were only known for the special case when the marked set consists of just a single vertex, or in the case of some specific graphs. We present a new quantum algorithm for finding a marked vertex in any graph, with any set of marked vertices, that is (up to a log factor) quadratically faster than the corresponding classical random walk. 
\end{abstract}

\section{Introduction}
	
As shown by Szegedy \cite{szegedy2004QMarkovChainSearch}, quantum walks provide a quadratic speedup over classical random walks for search tasks. If a classical random walk hits a marked element in an expected number of $HT$ steps, called the \emph{hitting time}, then the quantum walk runs in time $\bigObig{\sqrt{HT}}$. However, this speedup comes with a caveat: the quantum walk does not necessarily \emph{find} a marked element, but it can \emph{detect} a deviation from the starting state caused by marked elements. This issue has been well known since Szegedy's work in 2004, yet it has eluded all attempts to solve~it. 

Several generalizations of Szegedy's framework have been proposed but they only solve this issue in restricted cases. Tulsi \cite{tulsi2008FasterQW2DGrid} showed how to solve it for the random walk on an $N\times N$ grid with exactly one marked element. Here, the classical hitting time is $HT=\bigO{N^2 \log N}$. Szegedy's algorithm detects the presence of a marked element in $\bigObig{\sqrt{HT}}=\bigObig{N \sqrt{\log N}}$ steps. Measuring the final state of Szegedy's algorithm, however, only gives the marked element with probability $\Theta(1/\log N)$. Tulsi showed how to improve this to $\Theta(1)$, with the running time remaining $\bigObig{N \sqrt{\log N}}$. 
Magniez, Nayak, Richter and Santha \cite{magniez2012HittingTimesQWvsRW} then extended this to the random walk on any vertex transitive graph with exactly one marked element. Meanwhile Magniez, Nayak, Roland and Santha~\cite{magniez2006SearchQuantumWalk} presented an alternative extension of Szegedy's work, giving a quantum algorithm for finding a marked vertex that runs in a number of steps $\bigObig{\sqrt{1/(\delta\eps)}}$, where $\delta$ is the eigenvalue gap of (the Markov chain corresponding to) the walk and $\eps$ is the probability that a vertex is initially marked. This can be as small as $\bigObig{\sqrt{HT}}$ in certain cases, but significantly larger in others.

Later, Krovi, Magniez, Ozols and Roland \cite{krovi2010QWalkFindMarkedAnyGraph} proposed a new algorithm (based on a new notion of interpolated quantum walk) that achieves a quadratic advantage for finding a marked element for a random walk on any graph $G$ with exactly one marked element. The same result was achieved by Dohotaru and H{\o}yer \cite{dohotaru2017controlledQAmp}, using a different method. 

In the general case (with multiple marked elements), the algorithm of Krovi et al.\ finds a marked element, but takes time $\bigObig{\sqrt{HT^+}}$ where $HT^+$ is the {\em extended hitting time} of the walk. $HT^+$ is a new quantity obtained by modifying the expression for $HT$ in terms of eigenvalues and eigenvectors of the walk. If there is only one marked element, then $HT^+=HT$ and this yields the quadratic advantage for the quantum walk. However, $HT^+$ may be significantly larger than $HT$ when there are multiple marked elements,\footnote{The first version of the paper by Krovi et al.~\cite{krovi2010QWalkFindMarkedAnyGraph} claimed $HT^+=HT$ for any number of marked elements but this turned out to be false, as corrected by the authors in later versions.} as we show in Section~\ref{sec:2}.

Lastly, for a two-dimensional grid, a quadratic advantage for any set of marked elements was achieved by H{\o}yer and Komeili \cite{hoyer2017GridQWMultipleMarked} using a divide-and-conquer approach. However, their approach is specific to the two-dimensional grid and does not seem to generalize even to grids in higher dimensions.

In this paper, we finally resolve the problem of finding a marked element quadratically faster (up to a log factor) compared to the classical random walk, on any graph, for any number and any arrangement of marked elements.

First in Section~\ref{sec:2} we show that the gap between $HT^+$ and $HT$ can indeed be very large. We construct an arrangement of marked elements on an $N\times N$ grid for which $HT^+=\Omega(N^2)$ but $HT=\bigO{f(N)}$ where $f$ grows to infinity arbitrarily slowly. This shows that the algorithm of Krovi et al.\ can be severely suboptimal when there are multiple marked elements. The reason for this is that their algorithm actually solves a harder problem: it samples from the stationary distribution  restricted to marked vertices (which is the uniform distribution in case of the grid). Hence, their algorithm may be slow in cases when sampling from this distribution is substantially more difficult than just finding some marked element.

We then present two new algorithms in Section~\ref{sec:quantum-walk-alg}: a simpler algorithm, which we conjecture to find a marked element in time $\bigObig{\sqrt{HT}}$, for an arbitrary arrangement of marked elements (Conjecture \ref{conj:1}) and a more complicated algorithm for which we prove that it always finds a marked element in time $\widetilde{\mathcal{O}}\big(\sqrt{HT}\big)$ (Theorem~\ref{thm:main}). Both algorithms are based on the idea of interpolated walks, but use it differently from \cite{krovi2010QWalkFindMarkedAnyGraph}. 

The first algorithm, just runs the interpolated walk for $\bigObig{\sqrt{HT}}$ steps (instead of using eigenvalue estimation to produce an eigenstate of the walk, as in \cite{krovi2010QWalkFindMarkedAnyGraph}). Based on numerical experiments, we conjecture that, for any arrangement of marked vertices, there is a choice of the interpolation parameter and a choice of running time $t=\bigObig{\sqrt{HT}}$ which results in the walk producing a marked vertex with probability $\Omega(1)$. This conjecture holds for all the examples with $HT^+\gg HT$ that we could find, which we illustrate through some numerical experiments.

The second algorithm, combines the interpolated walk with the recently invented {\em quantum fast-forwarding} technique of Apers and Sarlette \cite{apers2018QFastForwardMarkovChains}. Quantum fast-forwarding is a primitive that allows one to replace $t$ steps of a classical random walk with $\bigO{\sqrt{t}}$ steps of a quantum walk, in a certain sense. A caveat is that quantum fast-forwarding may only produce the final state with a very small success probability. However, in our application, it succeeds with probability~$\widetilde{\Omega}(1)$. This is shown by an insightful argument that interprets the success probability of quantum fast-forwarding in terms of the classical random walk. Namely, it corresponds to the probability that the classical random walk, started in a random unmarked vertex, visits a marked vertex after $t$ steps, but returns to an unmarked vertex after $t$ additional steps. This probability can be tuned to be $\widetilde{\Omega}(1)$ by adjusting the interpolation parameter of the walk. 
	
\section{Preliminaries}

\subsection{Markov chains and random walks}

For a random variable $Z$ and probability distribution $\rho$, we will use $Z\sim \rho$ to indicate that $Z$ is distributed according to $\rho$.

A sequence of random variables $Y=(Y_i)_{i=0}^{\infty}$ is a Markov chain if for all $i>0$,
$$\Pr(Y_i=y_i|Y_0=y_0,\dots,Y_{i-1}=y_{i-1})=\Pr(Y_i=y_i|Y_{i-1}=y_{i-1}).$$ 
A (time-independent) Markov chain on a discrete state space $X$ with $|X|=n$ is specified by an $n\times n$ row-stochastic matrix $\PM$, whose $xy$-entry $\PM_{xy}$ denotes the probability that the Markov chain makes a transition from state $x\in X$ to the state $y\in X$ in one step. For a distribution $\rho$ on $X$, we say that $Y$ is a Markov chain evolving according to $\PM$ starting from $\rho$ if $Y_0\sim \rho$, and for all $i>0$ and $x,y\in X$, $\Pr(Y_i=y|Y_{i-1}=x) = \PM_{xy}$. We will left-multiply with probability (row) vectors to follow the common conventions in the literature for Markov chains, so if $Y_0\sim \rho$, then $Y_i\sim \rho\PM^i$, for any $i\geq 0$.

We say that $\PM$ is \emph{ergodic} if for a large enough $t\in \N$ all elements of $\PM^t$ are non-zero. For an ergodic $\PM$ there exists a unique stationary distribution $\uppi$ such that $\uppi \PM = \uppi$, and we define the \emph{time-reversed} Markov chain as $\PM^*:=\diag{\uppi}^{-1}\cdot\PM^T\cdot\diag{\uppi}$. We say that $\PM$ is \emph{reversible} if it is ergodic and $\PM^*=\PM$. Note that reversibility can be equivalently expressed by the \emph{detailed-balance} equations:
\begin{equation}\label{eq:balance}
\forall x,y \in X \colon \uppi_x \PM_{xy}=\uppi_y \PM_{yx},
\end{equation}
intuitively meaning that in the stationary distribution for each pair of states the probability of a transition between the states in both directions is that same. Moreover, it is easy to see that if $\PM$ is reversible then so is $\PM^t$ for every $t\in\N$.

For an ergodic Markov chain $\PM$, we define the discriminant matrix $D$ such that its $xy$-entry is $\sqrt{\PM_{xy}\PM^*_{yx}}$. It is easy to see that
\begin{equation}\label{eq:discriminant}
D=\diag{\uppi}^{\frac12}\cdot\PM\cdot\diag{\uppi}^{-\frac12}.
\end{equation}
This form has several important consequences. First of all the spectra of $\PM$ and $D$ coincide, and moreover, the vector $\sqrt{\uppi}$, that we get from $\uppi$ by taking the square root element-wise, is a left eigenvector of $D$ with eigenvalue $1$. Also from the definition $D_{xy}=\sqrt{\PM_{xy}\PM^*_{yx}}$ it follows that for reversible Markov chains, $D$ is a symmetric matrix, and therefore its singular values and eigenvalues coincide up to sign.

Reversible Markov chains are equivalent to random walks on weighted graphs; for a survey on the topic see Lovász~\cite{lovasz1993RandomWalksSurvey}. They have been used to design search algorithms in various contexts. Specifically, if $\PM$ is a random walk on a state space $X$, and $M\subset X$ is a set of \emph{marked} vertices, then a randomized algorithm that begins in any vertex $x\in X$ and repeatedly makes a step of the walk, while checking whether the current state is marked, will eventually find some $x\in M$ (assuming $M$ is non-empty). When the algorithm starts in the stationary distribution of $\PM$, the expected number of steps needed before a marked vertex is reached is called the \emph{hitting time}, and is denoted $HT=HT(\PM,M)$.  Let $Z$ be the smallest number such that $Y_Z\in M$, where $Y$ is a Markov chain evolving according to $\PM$ starting from $\uppi$, then $HT(\PM,M)=\mathbb{E}(Z)$. Moreover, by Markov's inequality, for any positive real number $c$ we have $\Pr(Z>cHT(\PM,M))\leq \frac{1}{c}$. 

Thus, for any reversible Markov chain $\PM$ on $X$, and $M\subset X$, if $\mathsf{C}$ is the complexity of checking whether $x\in M$ (for an arbitrary $x\in X$), $\sf U$ is the cost of taking one step of the walk $\PM$, and $\sf S$ is the cost of sampling according to the stationary distribution, then there is a randomized algorithm that finds a marked vertex with high probability in complexity $\bigO{{\sf S}+HT({\sf U}+{\sf C})}$. In the next subsection, we will consider quantum analogues of this procedure.  

For simplicity in the rest of the paper we will work with reversible time-independent Markov chains, unless otherwise stated.

\subsection{Interpolated walks and quantum walk search algorithms}

\paragraph{Interpolated walks.} Some previous quantum walk algorithms build on the notion of \emph{interpolated walk}. Intuitively speaking such a walk works as follows: first it checks whether the current node is marked. It the node is \emph{unmarked}, then it performs a normal step of the walk; but if it is \emph{marked}, then it performs a normal walk step only with probability $1-s$, and with probability $s$ it stays at the current marked node. 

Let us fix some reversible Markov chain $\PM$ and marked set $M\subset X$. 
We first define the \emph{absorbing} walk operator $\PM'$ as the modified Markov chain that, once it hits the set of marked vertices $M$, stays where it is.
If we arrange the states of $X$ so that the unmarked states $U := X \setminus M$ come first, matrices $\PM$ and $\PM'$ have the following block structure:
\begin{align*}
\PM :=\left(\begin{array}{cc} \PM_{UU} & \PM_{UM} \\ \PM_{MU} & \PM_{MM} \end{array}\right), & &
\PM' :=\left(\begin{array}{cc} \PM_{UU} & \PM_{UM} \\ 0 & I \end{array}\right).
\end{align*}
We define the \emph{interpolated walk} operator, for $s\in [0,1)$, as:
\begin{equation}\label{eq:interpolChainDef}
\PM(s):=(1 - s)\PM + s\PM',
\end{equation}
\emph{staying} at a marked vertex with probability $s$.
We denote the corresponding discriminant matrix by $D(s)$. Let $\Pi_M$ be the projector onto marked vertices and let $\Pi_U:=I-\Pi_M$ be the projector onto unmarked vertices. Then we define $\uppi_U:=\uppi \Pi_U$ and $\uppi_M:=\uppi \Pi_M$ as the row vectors that are obtained by restricting $\uppi$ to sets $U$ and  $M$, respectively. We denote the probability that an element is marked in the stationary distribution by $p_M:=\sum_{x\in M}\uppi_x$. Then $\uppi':=\uppi_M/p_M$ is a stationary distribution of $\PM'$.\footnote{In fact, any distribution with support only on marked states is stationary for $\PM'$.} In analogy to the definition of $\PM(s)$ in Eq.~\eqref{eq:interpolChainDef}, let $\uppi(s)$ be a convex combination of $\uppi$ and $\uppi'$, appropriately normalized:
\begin{equation}\label{eq:interpolDistDef}
\uppi(s):=\frac{(1 - s)\uppi + s\uppi'}{(1 - s) + sp_M}
=\frac{1}{1 - s(1-p_M)}((1 - s)\uppi_U + \uppi_M).
\end{equation}
Krovi et al.~\cite{krovi2010QWalkFindMarkedAnyGraph} showed that for any $s\in [0,1)$, $\PM(s)$ is a reversible ergodic Markov chain with unique stationary distribution $\uppi(s)$.

\paragraph{Quantum walk operator.} For a (reversible) Markov chain $\PM$, let $V(\PM)$ be a unitary such that\footnote{Note that here we swapped the role of the two registers compared to some previous works, in order to make the resemblance with block-encodings~\cite{chakraborty2018BlockMatrixPowers,gilyen2018QSingValTransf} more apparent, see Section~\ref{subsec:fast-forwarding} for more details.}
$$\forall x\in X\colon
V(\PM)\ket{\barO}\ket{x}=\sum_{y\in X}\sqrt{\PM_{xy}}\ket{y,x},$$
where $\ket{\barO}$ is some fixed reference state. 
The action of $V(\PM)$ is analogous to taking one step of the random walk $\PM$ in superposition.
Let $\textsc{Shift}$ be defined by the action
$\ket{x,y}\mapsto\ket{y,x},$
for all $x,y\in X$, and let $\textsc{Ref}= (2\ketbra{\barO}{\barO}-I)\otimes I$. The corresponding \emph{quantum walk operator} is
$$W(\PM):=V^\dagger\!(\PM)\, \textsc{Shift}\, V(\PM)\, \textsc{Ref}.$$
Note that $\bra{\barO}\bra{x}W(\PM)\ket{\barO}\ket{y}=\sqrt{\PM_{xy}\PM_{yx}}=D_{xy}$.

\paragraph{Extended hitting time.} For any $s\in [0,1)$, suppose $D(s)$ has eigenvalue decomposition $\sum_{k=1}^n\lambda_k(s)\ket{v_k(s)}\bra{v_k(s)}$, with $\lambda_n(s)=1$, so $\lambda_k(s)<1$ for all $k<n$. Then we can define 
$$HT(s)=\frac{1}{1-p_M}\sum_{k=1}^{n-1}\frac{|\braket{v_k(s)}{\sqrt{\uppi_U}}|^2}{1-\lambda_k(s)},
\quad\mbox{and}\quad
HT^+(\PM,M):=\lim_{s\rightarrow 1}HT(s),$$
where $\ket{\sqrt{\uppi_U}}=\sum_{x\in U}\sqrt{\uppi_x}\ket{x}$. We call $HT^+$ the \emph{extended hitting time}.
To put this definition into context, note that one can prove $HT(\PM,M)=\frac{1}{1-p_M}\sum_{k=1}^{n-|M|}\frac{|\braket{v_k'}{\sqrt{\uppi_U}}|^2}{1-\lambda_k'}$, where $\lambda_k'$ ranges over the ($\neq 1$) eigenvalues of $D(1)$ and $\ket{v_k'}$ are the corresponding eigenvectors. For a proof see, e.g., \cite[Proposition 9]{krovi2010QWalkFindMarkedAnyGraph}.

\paragraph{Quantum walk search algorithms.}
We introduce the following black-box operations:
\begin{itemize}
	\item $\check(M)$: checks if a given vertex is marked
	by mapping   $\ket{x} \ket{b}$ to $\ket{x} \ket{b}$ if $x \notin M$ and $\ket{x} \ket{b \oplus 1}$ if $x \in M$, where $\ket{x}$ is the vertex register and $b \in \lrb{0,1}$;
	\item $\setup(\PM )$: construct the superposition $\ket{\sqrt \uppi}  = \sum_{x \in X} \sqrt{\uppi_x} \ket{x}$;
	\item $\update(\PM)$: perform one update step. More precisely implement (separately, controlled versions of\footnote{This is mostly needed for implementing interpolated versions of the quantum walk.}) $\textsc{Shift}$, $\textsc{Ref}$, and $V(\PM)^{\pm1}$.
\end{itemize}
Each of these operations has a corresponding associated implementation cost, which we denote by $\checkingcost$, $\setupcost$, and $\updatecost$, respectively. 

For implementing the interpolated quantum walk we define a modified version of the update operator, which is a direct quantum analogue to the interpolated classical update: if the current vertex is marked flip a coin and do noting when the result is ``heads'', otherwise proceed as usually. Accordingly the modified quantum update operator $V(\PM,s)$ for all $x\in U$ acts as $I\otimes V(\PM)$ on the initial state $\ket{0}\ket{\barO}\ket{x}$, and for $x\in M$ acts as $\ket{0}\ket{\barO}\ket{x}\mapsto \sqrt{1-s}\ket{0}V(\PM)\ket{\barO}\ket{x} + \sqrt{s}\ket{1}\ket{\barO}\ket{x}$. We define the interpolated quantum walk operator as 
\begin{equation}
W(s):=V^\dagger\!(\PM,s) \,\textsc{Shift}'\, V(\PM,s)\, \textsc{Ref}',
\end{equation}
where $\textsc{Shift}':=\ketbra{0}{0}\otimes\textsc{Shift}+\ketbra{1}{1}\otimes I$ and $\textsc{Ref}':=(2\ketbra{0}{0}\otimes\ketbra{\barO}{\barO}-I)\otimes I$. It is easy to see that
\begin{equation}\label{eq:blockWalk}
\bra{0}\bra{\barO}\bra{x}W(s)\ket{0}\ket{\barO}\ket{y}=D_{xy}(s).
\end{equation}
Note that $W(s)$ can be implemented\footnote{We note that \cite[Appendix B.2]{krovi2010QWalkFindMarkedAnyGraph} also describes a way to implement the interpolated quantum walk operator with similar complexity but additionally require (query) access to the diagonal entries of $\PM$.} for any $s \in [0,1)$ in cost of order $\checkingcost+\updatecost$, the following way. First check whether $x\in X$ is marked, and if it is, then apply the map $\ket{0}\mapsto \sqrt{1-s}\ket{0} + \sqrt{s}\ket{1}$ to the first qubit. Controlled by the first qubit's state being $\ket{0}$ apply $V(\PM)$ to the last two registers.

While a classical random walk can find a marked vertex in complexity\footnote{We note that in the classical case, $\sf S$ can be replaced with the cost of \emph{classically} sampling from $\uppi$, and $\sf U$ with the cost of classically sampling a neighbour of the current vertex. These classical sampling operations may be cheaper than $\setup$ and $\update$, but in practice, they are often the same.} $\bigO{{\sf S}+HT({\sf U}+{\sf C})}$, Krovi et al.~\cite{krovi2010QWalkFindMarkedAnyGraph} showed that using the the quantum walk $W(s)$ one can find a marked vertex in complexity $\bigObig{{\sf S}+\sqrt{HT^+}({\sf U}+{\sf C})}$. In Section~\ref{sec:2}, we show that $HT^+$ may be much larger than $HT$, but then in Section~\ref{sec:fast-forwarding}, we show that in fact, a quantum algorithm can find a marked vertex in complexity $\bigOt{{\sf S}+\sqrt{HT}({\sf U}+{\sf C})}$, see Theorem~\ref{thm:main}. (From now on for simplicity we will just write $\ket{\barO}$ instead of $\ket{0}\ket{\barO}$ when we work with interpolated quantum walks $W(s)$.)

\section{Counterexample with \texorpdfstring{$HT^+\gg HT$}{HT+ >> HT}}\label{sec:2}

A torus is a graph containing $ n = N^2 $ vertices organized in $ N $ rows and $ N $ columns; there is a vertex $ (x_1,x_2) $ for all $ x_1,x_2\in \{0,1,\ldots, N-1\} $. A vertex $ (x_1,x_2) $ has four neighbours, $ (x_1+1,x_2) $, $ (x_1-1,x_2) $, $ (x_1,x_2+1) $ and $ (x_1,x_2-1) $, where the addition is modulo $ N $.
To prevent the graph from being bipartite, we add a self-loop at each vertex,  so that at any vertex the random walker moves to any of  the four neighbours with probability 0.2 and  stays at the same vertex  also with probability 0.2.

We start by observing that the extended hitting time $ HT^+ $ in the case of a torus can be lower bounded as follows. 
\begin{restatable}{lemma}{htpLB}\label{lem:htpLB}
	Let $ M  \subset  \lrb{0,1,\ldots,N-1}^2  $ be a set of marked   vertices of the $ N\times N $ torus.  Let $ m= \vert M \vert$, $ u = N^2 - m  $ and $ \omega =  \exp(2\pi \mathrm{i} / N) $. Then
	\vskip-3mm
	\begin{equation}\label{eq:L100}
	HT^+   \geq  \frac{5}{4} \frac{N^2}{m^2 u} \frac{\bigg|\vcenter{\hbox{$\sum^{\rotatebox[origin=c]{90}{\kern1.2mm}}\limits_{(x_1,x_2) \in M}  \omega^{x_1}$}}\bigg|^2}{\sin^2 \frac{\pi}{N}}.
	\end{equation}
\end{restatable}
\noindent The proof is deferred to Appendix \ref{app:A}.

Next we describe an example of a marked set whose extended hitting time  can be much larger than the hitting time.
\begin{lemma}\label{lem:htpGapHT}
	Suppose  that positive integers $ d_1, k_1, d, N $  satisfy the following requirements:   
	\begin{enumerate}[label=(C\arabic*),itemsep=0mm]
		\item \label{eq:S2C1} $ k_1 d_1 = \smallO{N} $;  
		\item \label{eq:S2C2} $ N = o(k_1 d) $;
		\item \label{eq:S2C3} $ d^2 \log d = \smallO{ N^2} $; 
		\item \label{eq:S2C4} $ d_1 $  is a divisor of $ d  $ and $ d  $ is a divisor of $ N $.
	\end{enumerate}
	Define a marked set $ M $ on the $ N \times N $ torus as $ M_1 \cup M_2 $, where
	\[ 
	M_1 = \lrb{ (j_1 d_1,  j_2 d_1)  \ \vline\    0 \leq  j_1, j_2 \leq k_1-1  }, 
	\quad
	M_2 = \lrb{ (j_1 d,  j_2 d)  \ \vline\    0 \leq  j_1, j_2 < N/d }.
	\]
	Then  the extended and classical hitting times for the set $ M $ satisfy
	\[  
	HT^+ = \Omega(N^2)
	\quad\text{and}\quad
	HT  = \bigO{d^2\log d}  = \smallO{HT^+},
	 \]
	 respectively.
\end{lemma}
In Figure \ref{fig:L1}   an illustration with $ d_1=1 $, $ k_1=15 $, $ d=6 $ and $ N=36 $ is depicted, with different colors for   $ M_1 \setminus M_2 $, $ M_2 \setminus M_1 $ and $ M_1 \cap M_2 $.

An example of   parameters satisfying \ref{eq:S2C1}-\ref{eq:S2C4}  is $ d_1  = 1 $, $ k_1 = a\,  2^{a^2}$,  $ d =  a^2 $ and $ N=a^2\, 2^{a^2} $, for an  integer $ a>1 $. For such parameters  Lemma~\ref{lem:htpGapHT} implies bounds $ HT = \bigO {\log^2 N \log \log N} $ and $ HT^+ = \Omega(N^2) $, thus there is a  $ {\widetilde{\Omega}\left( N^2 \right)} $ gap between  the extended hitting time $ HT^+ $ and  the classical hitting time $ HT $.

\begin{proof}[Proof of Lemma \ref{lem:htpGapHT}]
	
	Notice that the sets $ M_2 $ and $ M_1 $ overlap, since $ d_1 | d $ by \ref{eq:S2C4}.  The set $  M  $ consists of $ k_1 ^2 $ vertices forming a small, dense subgrid $ M_1 $, and the remaining marked vertices of $ M_2 $ forming a sparser subgrid in the rest of the torus.  

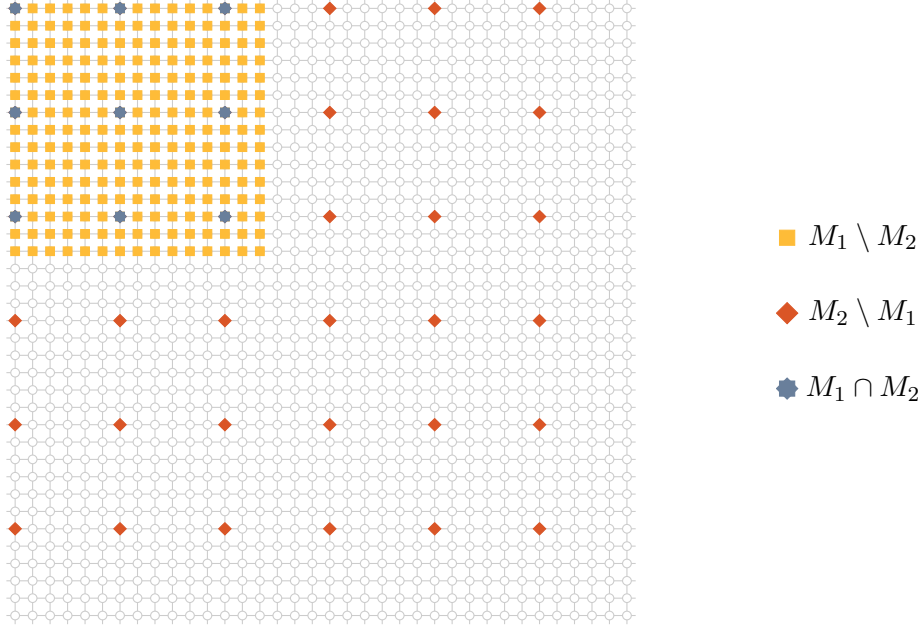
\begin{figure}[ht]
	\centering

	\definecolor{light-gray}{gray}{.8}
	\definecolor{reddish}{RGB}{218,85,38}
	\definecolor{yellowish}{RGB}{254,188,56}
	\definecolor{blueish}{RGB}{105,127,155}

	%%% First version of the grid drawing: %%%
	%
	%\begin{tikzpicture}[scale=.15]
	%\filldraw[color=yellowish] (0,35) rectangle (14,21);
	%
	%\foreach \x in {0,6,...,35}{
	%	\foreach \y in {0,6,...,35}{
	%		\filldraw[color=reddish] (\x,35-\y) rectangle (\x+1,34-\y);
	%	}
	%}
	%
	%\foreach \x in {0,6,...,12}{
	%	\foreach \y in {0,6,...,12}{
	%		\filldraw[color=blueish] (\x,35-\y) rectangle (\x+1,34-\y);
	%	}
	%}
	%
	%\foreach \x in {0, 1, ..., 35} \draw[color=gray](\x,0)--(\x,35);
	%\foreach \x in {0, 1, ..., 35} \draw[color=gray](0,\x)--(35,\x);
	%
	%\end{tikzpicture}

	\begin{tikzpicture}
	\node at (0,0) {
	\begin{tikzpicture}[scale=.23]
	
	\foreach \x in {0, 1, ..., 35} \draw[color=light-gray](\x,-.5)--(\x,35.5);
	\foreach \y in {0, 1, ..., 35} \draw[color=light-gray](-.5,\y)--(35.5,\y);
	
	\foreach \x in {0,1,...,35}{
		\foreach \y in {0,1,...,35}{
			\filldraw[color=light-gray,fill=white] (\x,\y) circle (.25);
		}
	}
	
	\foreach \x in {0,1,...,14}{
		\foreach \y in {21,22,...,35}{
			\filldraw[color=yellowish] (\x-.25,\y-.25) rectangle (\x+.25,\y+.25);;
		}
	}
	%\filldraw[color=green] (0,35) rectangle (14,21);
	
	\foreach \x in {0,6,...,35}{
		\foreach \y in {5,11,...,35}{
			\filldraw[color=reddish] (\x,\y+.354) -- (\x+.354,\y) -- (\x,\y-.354) -- (\x-.354,\y) -- (\x,\y+.354);%(\x,\y) circle (.25);
		}
	}
	
	\foreach \x in {0,6,12}{
		\foreach \y in {23,29,35}{
			\filldraw[color=blueish] (\x,\y+.354) -- (\x+.354,\y) -- (\x,\y-.354) -- (\x-.354,\y) -- (\x,\y+.354);
			\filldraw[color=blueish] (\x-.25,\y-.25) rectangle (\x+.25,\y+.25); %(\x,\y) circle (.25);
		}
	}
	\end{tikzpicture}};
	\node at (7,0){\begin{tikzpicture}
	\filldraw[color=yellowish] (-.1,1.9) rectangle (.1,2.1);
	\filldraw[color=reddish] (-.14,1)--(0,1.14)--(.14,1)--(0,.86)--(-.14,1); 
	\filldraw[color=blueish] (-.1,-.1) rectangle (.1,.1);
	\filldraw[color=blueish] (-.14,0)--(0,.14)--(.14,0)--(0,-.14)--(-.14,0);
	\node at (1,2) {$M_1\setminus M_2$};
	\node at (1,1) {$M_2\setminus M_1$};
	\node at (1,0) {$M_1\cap M_2$};
	\end{tikzpicture}};
	\end{tikzpicture}

	\caption{Illustration of the marked set with $ d_1=1 $, $ k_1=15 $, $ d=6 $ and $ N=36 $.}
	\label{fig:L1}
\end{figure}
	
	Since $ m = \lrv M \leq \lrv {M_1} + \lrv {M_2}  = k_1 ^2  + (N/d)^2 $, the constraint \ref{eq:S2C2} implies $ m = \bigO{k_1^2} $;  from \ref{eq:S2C1}  we conclude   $ m = o(N^2) $ and $ N^2 / u = \Theta(1) $. Moreover,  $ \sin^2 \frac{\pi}{N} \sim  \frac{\pi^2}{N^2} = \Theta(N^{-2}) $, thus \eqref{eq:L100} gives 
	\begin{equation}\label{eq:S2e01}
	HT^+  = \Omega  \lr{N^ 2  \lrv \rho ^2/m^2  },
	%\quad\text{where}\quad
	% \rho = \sum\limits_{(x_1,x_2) \in M}  \omega^{x_1} .
	\end{equation}
	where $\rho$ is defined by
	\[ 
	\rho  =
	\sum_{x\in M}  \omega^{x_1}  = 
	\sum_{x\in M_1}  \omega^{x_1}  + \sum_{x\in M_2}  \omega^{x_1}  - \sum_{x\in M_1 \cap M_2}  \omega^{x_1}. 
	\]
	The first summand on the RHS is
	\[ 
	\sum_{x\in M_1}  \omega^{x_1}  = k_1  \sum_{j=0}^{k_1-1} \omega^{j d_1} = k_1 \frac{\omega^{k_1 d_1}-1}{\omega^{d_1}-1},
	\]
	while the second summand  is a multiple of $ \sum_{j=0}^{N/d-1} \omega^{j d} =   \lr{\omega^{N}-1}/\lr{\omega^{d}-1} =0 $ because $ d|N $ by \ref{eq:S2C4}. Therefore
	\[ 
	\rho = k_1 \frac{\omega^{k_1 d_1}-1}{\omega^{d_1}-1} - \sum_{x\in M_1 \cap M_2}  \omega^{x_1} .
	\]
	It is easy to see that $ M_1 \cap M_2=  \lrb{ (j_1 d,  j_2 d)  \ \vline\    0 \leq  j_1, j_2 <k } $, where $ k= \lceil k_1 d_1/d  \rceil $, and similar arguments as previously yield
	\[ 
	\rho = k_1 \frac{\omega^{k_1 d_1}-1}{\omega^{d_1}-1} -  k \frac{\omega^{ k d}-1}{\omega^{d}-1}.
	\]
	By the reverse triangle inequality,
	\begin{equation}\label{eq:S2e02}
	\lrv{\rho} \geq  k_1  \frac{\lrv{\omega^{k_1 d_1}-1  }}{\lrv{\omega^{d_1}-1}}  - k   \frac{\lrv{\omega^{kd }-1  }}{\lrv{\omega^{d}-1}}
	=
	\frac{k_1  \sin \frac{\pi k_1 d_1}{N}}{\sin \frac{\pi  d_1}{N}}
	-
	\frac{k  \sin \frac{\pi k d}{N}}{\sin \frac{\pi  d}{N}}.
	\end{equation}
	From \ref{eq:S2C1} and  \ref{eq:S2C3}   we obtain $ kd \leq k_1d_1+d = \smallO{N} $, therefore $ \frac{k_1d_1}{N} = \smallO 1$, $ \frac{kd}{N} = \smallO 1$ and 
	\[ 
	\sin \frac{\pi k_1 d_1}{N} = \Theta \lr{ \frac{k_1 d_1}{N} }, \
	\sin \frac{\pi k d}{N} = \Theta \lr{ \frac{kd}{N} }, \
	\sin \frac{\pi  d_1}{N} = \Theta \lr{\frac{d_1}{N}},\ 
	\sin \frac{\pi  d}{N} = \Theta \lr{\frac{d}{N}}.
	\] 
	Consequently,
	\[ 
	\frac{k_1  \sin \frac{\pi k_1 d_1}{N}}{\sin \frac{\pi  d_1}{N}} = \Theta(k_1^2), \quad
	\frac{k  \sin \frac{\pi k d}{N}}{\sin \frac{\pi  d}{N}} = \Theta(k^2)  =  \Theta\lr{k_1^2 \frac{d_1^2}{d^2} } = \smallO{k_1^2};
	\]
	here the last bound follows from $ d_1=\smallO{d} $,  which  is implied by \ref{eq:S2C1} and \ref{eq:S2C2}.
	
	Now \eqref{eq:S2e02} gives $ \lrv \rho  = \Omega(k_1^2) $.  Combining this with \eqref{eq:S2e01} and the previously  obtained bound  $ m= \bigO{k_1^2} $, we conclude that the  extended hitting time satisfies
	\begin{equation}\label{eq:S2e03}
	HT^+ = \Omega\lr{ \frac{N^ 2  \lrv \rho ^2  }{m^2}}  = \Omega \lr{ \frac{N^2 k_1^4}{k_1^4} } = \Omega(N^2).
	\end{equation}
	
	Next we   bound $ HT  $.  Notice that by the linearity of expectation $ HT = \sum_{x\in U} \frac{\pi_x}{p_U} HT_x (M)$, where  $ p_U := \sum_{x\in U} \pi_x $ and $ HT_x (M)$ is the expected number of steps for the random walker to reach $ M $ for the first time, starting from a vertex $ x $. It follows that $ HT \leq \max_{x \in U}HT_x (M)$. For any fixed $ x\in U $,  $ HT_x (M)$ cannot decrease when reducing the marked set (i.e., when some marked vertices are removed from $ M $ and added to  the unmarked set $ U $), hence we have $ \max_{x \in U} HT_x(M) \leq  \max_{x \notin M_2} HT_x(M_2) $.
	
	Therefore it suffices to show that $   HT_x(M_2)   = \bigO{d^2\log d}  $  when only the subgrid $ M_2 $ is marked and $ x  $ is any vertex not belonging to $ M_2 $. 
	However,  the classical random walk with the marked set $   M_2 $  is equivalent to the random walk in the $ d\times d $ torus with a single marked element (by identifying each vertex $ (x_1,x_2) $ with the unique vertex $ (x_1^{(0)},x_2^{(0)})  $ satisfying $ x_1 \equiv x_1^{(0)}   \pmod d  $, $ x_2 \equiv x_2^{(0)} \pmod d $ and  $ x_1^{(0)}, x_2^{(0)} \in \lrb{0,1,\ldots, d-1}  $).
	Since, in the case of a $ d\times d $  torus with a single marked element, all hitting times $ HT_y $ (with $ y $ being a non-marked vertex) are  of order $ \bigO{d^2 \log d} $ \cite[Eq.~10.29]{levin2017MarkovChainsMixingTimes}, the desired bound $ HT_x(M_2)   = \bigO{d^2\log d}  $ follows.  Hence, returning to the marked set $ M $,   the classical hitting time is $ HT =  \bigO{d^2 \log d} = \smallO{N^2} $ by~\ref{eq:S2C3}, and we conclude that  $ HT = \smallO{HT^+} $.
\end{proof}

An intuitive explanation for this result is that the algorithm of Krovi et al.~\cite{krovi2010QWalkFindMarkedAnyGraph} actually solves a more difficult problem: it generates the uniform superposition over $\ket{x}, x\in M$ (with the starting state being the uniform superposition over all vertices $\ket{x}$). Almost all of marked vertices are, however, concentrated in $M_1$ which is a small part of the grid. A typical component of the starting state is at a distance $\Omega(N)$ from $M_1$. Therefore, any algorithm that generates the uniform superposition over $\ket{x}, x\in M$ from this starting state must take $\Omega(N)$ steps, even though the classical hitting $HT$ time is much smaller. 

The running time $O(\sqrt{HT^+})=O(N \sqrt{\log N})$ achieved by the algorithm of \cite{krovi2010QWalkFindMarkedAnyGraph}
is quite close to the $\Omega(N)$ lower bound. So, in our example, this algorithm is close to being optimal for generating the uniform superposition of marked vertices but is very far from being optimal for the task of 
simply finding a marked vertex.

\section{Quantum Walk Algorithm}\label{sec:quantum-walk-alg}
As in \cite{krovi2010QWalkFindMarkedAnyGraph}, we introduce the $  n$-dimensional Hilbert space $\Hi $  with basis states $ \ket x $ identified with the vertices of the graph. The algorithm uses two registers $\Reg_1$, $\Reg_2$ with underlying state space $\Hi$ for each of them,  initialized to some reference state $ \ket{\barO} $.

Additionally an ancilla register $\Reg_3$ initialized to $\ket{0} \in \C^2$ will be attached to check if the current vertex is marked.

\subsection{Algorithm with known \texorpdfstring{$s$ and $t$}{s and t}}
Now we describe a quantum walk algorithm with a fixed interpolation parameter  $ s\in [0,1) $ and a predetermined  number of   quantum walk steps $ t \in \N $.

\begin{algorithm}[H]
	\textbf{Search}($ \PM $, $ M $, $ s $, $ t $) %: aim
	
	\begin{enumerate}
		\item Prepare the state $\ket{\barO}\ket{\sqrt \uppi}$ with $ \setup(\PM) $.
		\item Apply $ t $ times  the operator $ W(s) $  on $ \Reg_1 \Reg_2 $.
		\item Attach $ \Reg_3 $, apply $ \check(M) $ on $ \Reg_2 \Reg_3 $, measure $ \Reg_3 $.
		\item If $ \Reg_3=1 $, measure $ \Reg_2 $ in the vertex basis, output the outcome. Otherwise, output \texttt{No marked vertex found}.
	\end{enumerate}
	
	\caption{Quantum walk algorithm}\label{alg:alg1}
\end{algorithm}

It is obvious  that  the   complexity  of the algorithm is of the order $\setupcost + t \cdot (\updatecost + \checkingcost)$. 
We conjecture that (under the   assumption that the probability to draw a marked vertex from the stationary distribution is at most 0.5) there always exists an interpolation parameter $s$ such that Algorithm~\ref{alg:alg1} finds a marked vertex with high probability in $ t =\bigObig{\sqrt{HT}} $ steps:
\begin{conj}
\label{conj:1}
	Let  $  \PM$ be a reversible, ergodic Markov chain with stationary distribution $\uppi$;  suppose that   $ M $ is a set of marked states which  satisfies $p_M= \sum_{x\in M} \uppi_x < 0.5 $. Then
	there exists a  value $ s \in [0,1) $ and a positive  integer $ t = \bigObig{\sqrt{HT}} $ such that
	Algorithm \ref{alg:alg1} succeeds with probability $ \Omega(1) $.
	%	$ q_t(s) = \Omega(1) $, where $ q_t(s) $ is defined by \eqref{eq:S3e01}.
\end{conj}

The success probability  can be lower-bounded by a quantity expressible in terms of the discriminant matrix $ D(s) $.
Let $ \psuccess = \nrm{ (I \otimes \Pi_M) W^t(s)  \ket{\barO} \ket{\sqrt \uppi} }^2 $ be the probability of obtaining a marked vertex in the last step of Algorithm \ref{alg:alg1}. Then it  can be lower-bounded by
\begin{equation}\label{eq:S3e01}
\nrm{ (I \otimes \Pi_M)  \Pi_0 W^t(s)  \ket{\barO} \ket{\sqrt \uppi} }^2
=: q_t(s),
\end{equation}
where $\Pi_0 :=  \ketbra{\barO}{\barO} \otimes I $.   The following lemma\footnote{For a generalization of this claim see \cite[Lemma 9 \& Theorem 17]{gilyen2018QSingValTransf}.} implies that  $q_t(s)=  \nrm{ \Pi_M D_t(s)  \ket{\sqrt \uppi} }^2$, where $D_t(s)=T_t(D(s))$ for $T_t$ the Chebyshev polynomial of the first kind of degree $t$.
\begin{lemma}\label{lem:L2}
	The quantum walk operator $ W^t(s) $, when restricted to $\ket{\barO}$ in the first register, acts as the $t$-th Chebyshev polynomial of the first kind applied to the discriminant matrix $D(s)$, i.e.,
	\[ 
	\Pi_0  W^t(s)  \Pi_0 =  \ketbra{\barO}{\barO} \otimes  D_t(s) ,
	 \]
	 where $ D_t(s) = T_t(  D(s) ) $ and $ T_t $ is the  Chebyshev  polynomial of the first kind of degree $ t $, applied (in the matrix function sense) to the matrix $ D(s) $. Equivalently, $ D_t(s) $ can be defined via the recurrence relations
	 \begin{align}
	 &D_0(s)  = I, \quad  D_1(s) = D(s) ,   \label{eq:L201} \\
	 & D_{t+1}(s)  = 2 D_{t}(s) \cdot D(s) - D_{t-1} (s), \quad t \in \N . \label{eq:L202}
	 \end{align}
\end{lemma}
\begin{proof}
	Recall that $ W(s)  = \widetilde  W(s) \cdot (2\Pi_0 - I \otimes I) $ where $ \widetilde   W(s) =V^\dagger\!(\PM, s)\, \textsc{Shift}'  \,V(\PM,s)$. Moreover, the idempotence of $ \Pi_0 $ gives  
	\begin{equation}\label{eq:L203}
	W(s) \Pi_0  =  \widetilde  W(s) \cdot (2\Pi_0 - I \otimes I) \Pi_0 = \widetilde  W(s) \Pi_0.
	\end{equation}

	For the proof by induction on $ t $, notice that the claim trivially holds for $ t=0 $.  When $ t=1 $, the statement (due to \eqref{eq:L203}) is equivalent to Eq.~\eqref{eq:blockWalk}. Suppose that the claim has been proven for all nonnegative integers up to $ t $ inclusive, $ t\geq 1 $, and consider $ \Pi_0  W^{t+1}(s)  \Pi_0 $.
	We have
	\begin{align*}
	 \Pi_0  W^{t+1}(s)  \Pi_0 
	&  =  \Pi_0  W^{t-1}(s) \cdot \lr{ \widetilde W(s) \cdot (2\Pi_0 - I \otimes I)  }  W(s)\Pi_0  \\
	&  =    2  \Pi_0  W^{t-1}(s)  \widetilde W(s)  \Pi_0   W(s)\Pi_0  - \Pi_0  W^{t-1}(s)  \widetilde W(s)      W(s)\Pi_0 \\
	 & =     2  \Pi_0  W^{t-1}(s)    W(s)  \Pi_0   W(s)\Pi_0  - \Pi_0  W^{t-1}(s)  \widetilde W(s)     \widetilde W(s)\Pi_0  \tag{by Eq.~\eqref{eq:L203}}\\
	 & = 2  \Pi_0  W^t(s)  \Pi_0 \cdot   \Pi_0   W(s)\Pi_0  - \Pi_0  W^{t-1}(s) \Pi_0. \tag{since $\widetilde W^2(s) = I$ and $ \Pi_0^2 = \Pi_0 $}
	\end{align*}
	%where the last equality uses $ \widetilde W^2(s) = I $ and the idempotence of $ \Pi_0 $. 
	By the  inductive hypothesis, the obtained quantity equals $ \ketbra{\barO}{\barO} \otimes \lr{ 2 D_{t}(s) \cdot D_1(s) -  D_{t-1}(s)}  $. We conclude  that indeed $  \Pi_0  W^{t+1}(s)  \Pi_0 =\ketbra{\barO}{\barO} \otimes D_{t+1} (s) $, where $ D_{t+1}(s) $ is defined by the recurrence relations \eqref{eq:L201}-\eqref{eq:L202}. It remains to recognize that these recurrence relations define the Chebyshev polynomials of the first kind. 
\end{proof}

In the following we describe some examples illustrating  the dependence of $ q_t(s) $ on the interpolation parameter $ s $.

\begin{exmp} \label{exmp:3.1}
	Consider the example described in Section \ref{sec:2}, with parameter $ a=3 $ (i.e., $ d_1=1 $, $ k_1=1536 $, $ d=9 $,  and $ N=4608 $). It can be calculated that the classical hitting time of the marked set is $ HT =162.98\ldots  $, whereas the extended hitting time is  $HT^+ = 1.01\ldots \cdot 10^{7} $    (the lower bound in Lemma \ref{lem:htpGapHT}  gives $ HT^+ \geq  1.69 \cdot 10^6  $, by \eqref{eq:L100} and \eqref{eq:S2e02}).
	
	In Figure \ref{fig:S3.2-1}, we plot  the lower bound (\ref{eq:S3e01}) on the success probability of Algorithm~\ref{alg:alg1}. As we will also see in  Section~\ref{sec:fast-forwarding}, it is natural to replace the interpolation parameter $ s\in[0,1) $ with $ r = 1/(1-s)  \in [1,\infty)$. (The parameter $r$ is also equal to the expected number of steps until the interpolated walk makes a transition according to the original random walk at a marked vertex.) 
	
	\begin{figure}[ht]
		\centering
		\includegraphics[draft=false,scale=1]{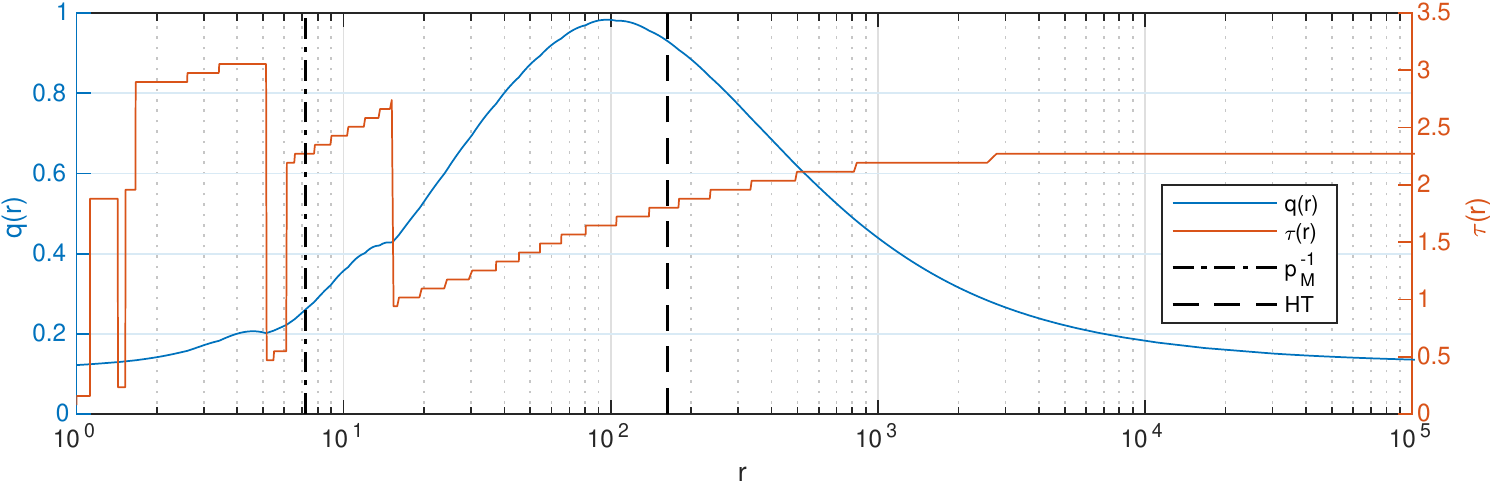}
		\caption{Bounds on Algorithm \ref{alg:alg1} in Example~\ref{exmp:3.1}. The horizontal axis $(r)$ represents the interpolation parameter $s=1-\frac{1}{r}$;
		$\tau(r)$ denotes the best choice of time $t$ and $q(r)$ denotes the best lower bound on the success probability of Algorithm~\ref{alg:alg1}, as described below.}
		\label{fig:S3.2-1}
	\end{figure}
	
	Figure \ref{fig:S3.2-1} shows two quantities (as functions of $r$):
	\begin{itemize}
	    \item 
	    the maximum of the bound (\ref{eq:S3e01}) over $t \leq   \lceil  3 \sqrt{HT}\rceil$,  denoted $ q(r)$ (units on the left axis);
	    \item
	    the minimal value of $ t $ which achieves $  q(r) $, 
	    denoted by $ \tau(r)\! := \min  \lrb{t \geq 0  \ \vline\  q_t(1\!-\!\frac1r)\! =\! q(r) } $  (with units on the right axis; represented in $\sqrt{HT}$ units).
	\end{itemize}

	Furthermore, we indicate parameter values $ r _1= \frac{1-p_M}{p_M} $ (which corresponds to the value of $s$ used in \cite{krovi2010QWalkFindMarkedAnyGraph} for their $ \Theta \big(\sqrt{HT ^+}\big)  $-time algorithm) and $ r_2=HT $ (a plausible upper bound on the optimal $r$) by vertical dash-dotted and dashed lines, respectively.
	
    From Figure~\ref{fig:S3.2-1} it can be noticed that  the optimal value is $ r =96.61\ldots  \approx d^2 $ and it allows Algorithm \ref{alg:alg1} to find a marked vertex in  $t=21 \approx 1.65 \sqrt {HT}$ steps with probability exceeding 0.98.
    This value is substantially bigger than the value $r_1 \approx 7.191$ 
    corresponding to the algorithm of \cite{krovi2010QWalkFindMarkedAnyGraph}.
\end{exmp}

\begin{figure}[ht]
	\centering
	\vskip-2mm
	\includegraphics[draft=false,scale=1]{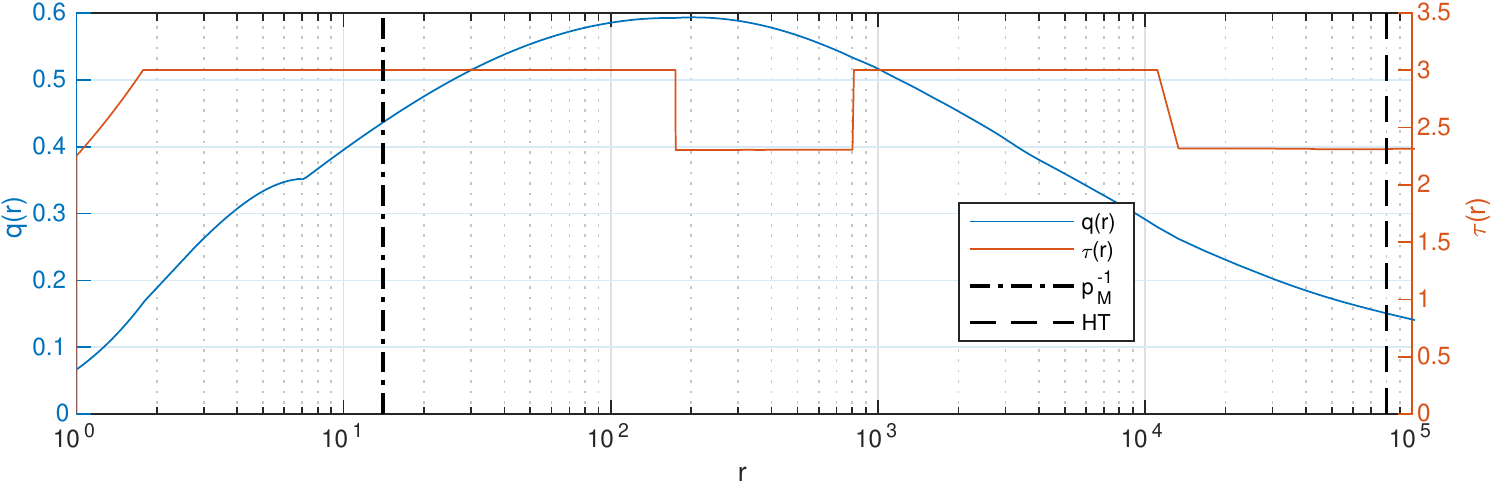}
	\caption{Some properties of the family of interpolated quantum walks in Example~\ref{exmp:3.2}. For notation and explanation of the plotted quantities see Figure~\ref{fig:S3.2-1}.}		
	\label{fig:S3.2-2}
\end{figure}
\begin{exmp} \label{exmp:3.2}
	Let $ G_k $ be the graph   consisting of a single central node $ x_0 $ and  $ k $ paths of length $ k^2 $;   all  paths  have a common endpoint $ x_0 $ and the remaining vertices are distinct	(i.e., $ G_k $ is a modified version of the  star graph with $ k $ rays of length $ k^2 $). In each vertex the random walker stays in the same vertex with probability 0.5 and with probability 0.5 moves to a neighbour vertex (in case of several neighbours, the probability 0.5   splits evenly among them to move to a particular neighbour). Let $ M $ be one of the $ k  $ paths, not including the central node.
	
	When $ k=15 $, the classical hitting time is $ HT =  80090.95\ldots$, whereas the extended hitting time is $ HT^+ =  1016848.98\ldots $.
	As previously, we change variables $ r = 1/(1-s) $ and plot	$ q(r)$ and	$\tau(r)$ on the left and right axis of Figure \ref{fig:S3.2-2}, respectively. 
	Again, values $ r _1= \frac{1-p_M}{p_M} $ and $ r_2=HT $ are indicated by vertical  lines.
	As indicated by Figure \ref{fig:S3.2-2},  at $r\approx k^2  $  Algorithm \ref{alg:alg1} finds a marked vertex with probability at least $ 0.59$ in less than $ 2.31 \sqrt{HT} $ steps.
\end{exmp}

\section{Fast-forwarding Algorithm}\label{sec:fast-forwarding}

In this section, we prove our main theorem, which is the following.

\begin{theorem}\label{thm:main}
Let $\PM$ be any reversible Markov chain on a finite state space $X$, and let $M\subset X$ be a marked set. There is a quantum algorithm that outputs a vertex $x$ from $M$ with bounded error in complexity
$$\bigO{\mathsf{S}\sqrt{\log(HT)}+\sqrt{HT}({\sf U}+{\sf C})\sqrt{\log(HT)\log\log(HT)}},$$
where $HT$ is a known upper bound on $HT(\PM,M)$, $\sf S$ is the cost of the $\mathtt{Setup}(\PM)$ operation, $\sf U$ is the cost of the $\mathtt{Update}(\PM)$ operation, and $\sf C$ is the cost of the $\mathtt{Check}(M)$ operation. 
\end{theorem}

We remark that if no upper bound on $HT(\PM,M)$ is known, then we can apply the exponential search algorithm of Boyer, Brassard, H{\o}yer and Tapp \cite{boyer1998TightBoundsOnQuantumSearching}, where we simply run the algorithm from Theorem~\ref{thm:main} with exponentially increasing guesses of an upper bound $HT$. 
This leads to the following corollary.
\begin{corollary}
Let $\PM$ be any reversible Markov chain on a finite state space $X$, and let $M\subset X$ be a marked set. There is a quantum algorithm that outputs a vertex $x$ from $M$ with bounded error in expected complexity
$$\bigO{\mathsf{S}\log^{1.5}(HT)+\sqrt{HT}({\sf U}+{\sf C})\sqrt{\log(HT)\log\log(HT)}},$$
where $HT=HT(\PM,M)$, $\sf S$ is the cost of the $\mathtt{Setup}(\PM)$ operation, $\sf U$ is the cost of the $\mathtt{Update}(\PM)$ operation, and $\sf C$ is the cost of the $\mathtt{Check}(M)$ operation. 
\end{corollary}

\subsection{Quantum fast-forwarding}\label{subsec:fast-forwarding}
We will use the quantum fast-forwarding technique of Apers and Sarlette \cite{apers2018QFastForwardMarkovChains}, which allows us to, in some very ``quantum'' sense, apply $t$ steps of a walk in only $\sqrt{t}$ calls to its update operation. We invoke their main result and state it in a slightly adapted form.

\begin{theorem}[\cite{apers2018QFastForwardMarkovChains}]\label{thm:fast-forwarding}
Let $\eps\in(0,1)$, $s\in[0,1]$ and $t\in\mathbb{N}$. Let $\cal P$ be any reversible Markov chain on state space $X$, and let $\mathsf{Q}$ be the cost of implementing the (controlled) quantum walk operator $W(s)$. There is a quantum algorithm with complexity $\bigO{\sqrt{t\log(1/\eps)}\mathsf{Q}}$ that takes input $\ket{\barO}\ket{\psi}\in \mathrm{span}\{\ket{\barO}\ket{x}:x\in X\}$, and outputs a state that is $\eps$-close to a state of the form
$$\ket{0}^{\!\otimes a}\ket{\barO}D^t\ket{\psi}+\ket{\Gamma}$$
where $a=\bigO{\log(t\log(1/\eps))}$ and $\ket{\Gamma}$ is some garbage state that has no support on states containing $\ket{0}^{\!\otimes a}\ket{\barO}$ in the first two registers.%, and $\nrm{\ket{\Gamma}}^2 = 1 - \nrm{D^t\ket{\psi}}^2$ (which may depend on $\ket{\psi}$ and $t$). 
\end{theorem}

To gain some intuition it is useful to think about the $W$ walk operator as a block-encoding of the discriminant matrix $D$, i.e., a unitary matrix containing $D$ in the top-left corner. In this terminology, fast-forwarding reads as implementing a block-encoding of $D^t$ by using the block-encoding of $D$ only $\propto\sqrt{t}$ times. By this insight one can rederive Theorem~\ref{thm:fast-forwarding} via recent qubitization~\cite{low2017HamSimUnifAmp} or quantum singular value transformation~\cite{gilyen2018QSingValTransf}~result as well.

Consider the case when we start with the subnormalized vector $\ket{\sqrt{\uppi_U}}=\sum_{x\in U}\sqrt{\uppi_x}\ket{x}$ and apply the ``fast-forwarded'' Markov chain from Theorem~\ref{thm:fast-forwarding}, before measuring. We show how to re-express the probability of measuring a marked element in terms of the interpolated walk $\PM(s)$. The probability of measuring a marked state is given by the square of:\footnote{For a parametrized matrix $M(s)$ we denote $(M(s))^t$ simply by $M^t(s)$, so for example $\PM^t(s)\equiv (\PM(s))^t$.}
\begin{align}
\nrm{\Pi_M D^t(s)\ket{\sqrt{\uppi_U}}} &\geq \nrm{\Pi_M D^t(s)\ket{\sqrt{\uppi_U}}}\nrm{\Pi_M D^{\hat{t}}(s)\ket{\sqrt{\uppi_U}}}\tag*{$\,\forall \hat{t}$, since $\nrm{D(s)}=1$}\\
&\geq \bra{\sqrt{\uppi_U}}D^t(s)\Pi_M D^{\hat{t}}(s)\ket{\sqrt{\uppi_U}} \tag*{by Cauchy-Schwarz}\\
&=\bra{\sqrt{\uppi_U}}\diag{\uppi(s)}^{\frac12}\PM^t(s)\Pi_M  \PM^{\hat{t}}(s)\diag{\uppi(s)}^{-\frac12}\ket{\sqrt{\uppi_U}} \tag*{by Eq.~\eqref{eq:discriminant}}\\
&\overset{}{=}\sum_{x,z\in U}\uppi_x\bra{x}\PM^t(s)\Pi_M\PM^{\hat{t}}(s)\ket{z}.\label{eq:probMUPrime}
\end{align}
In the last equality, we have used the fact that, from Eq.~\eqref{eq:interpolDistDef}, $\uppi(s)$ restricted to $U$ is proportional to $\uppi$, so for some $\alpha$, $\bra{\sqrt{\uppi_U}}\diag{\uppi(s)}^{\frac12}=\bra{\sqrt{\uppi_U}}\sqrt{\alpha}\diag{\uppi_U}^{\frac12}$, and $\diag{\uppi(s)}^{-\frac12}\ket{\sqrt{\uppi_U}}=\frac{1}{\sqrt{\alpha}}\diag{\uppi_U}^{-\frac12}\ket{\sqrt{\uppi_U}}$. 

The expression in \eqref{eq:probMUPrime}, equivalently expressed as $\nrm{\bra{\uppi_U}\PM^t(s)\Pi_M\PM^{\hat{t}}(s)\Pi_U}_1$, is the probability that upon starting from the stationary distribution of $\PM$ and evolving according to $\PM(s)$, the first vertex is unmarked, after $t$ steps we are at a marked vertex, and after another $\hat{t}$ steps we are at an unmarked vertex again. We summarize this in the following lemma:
\begin{lemma}\label{lem:probMUPrime}
Let $s\in [0,1)$, and $\PM$ be any reversible Markov process. Let $Y(s)=(Y_i(s))_{i=0}^{\infty}$ be the Markov chain evolving according to $\PM(s)$ starting from $Y_0(s)\sim \uppi$. Then for any $t,\hat{t}\in\mathbb{N}$, letting $t'=t+\hat{t}$:
\begin{equation}
\nrm{\Pi_M D^t(s)\ket{\sqrt{\uppi_U}}}\geq \Pr\left(Y_0(s)\in U,Y_t(s)\in M,Y_{t'}(s)\in U\right).\label{eq:probMUY}
\end{equation}
\end{lemma}
Thus, it suffices to lower bound the probability in \eqref{eq:probMUY} by $\widetilde{\Omega}(1)$ for some choice of $s$ and $t=\bigO{HT}$. Note that $t'>t$ can be arbitrarily large. In the next section, we lower bound \eqref{eq:probMUY}.

\subsection{Combinatorial Lemma}
To lower bound \eqref{eq:probMUY} by $\widetilde{\Omega}(1)$, we want to prove that (for some $s$), if we start in the stationary distribution and run the chain, there is some random choice of $t,t'=\bigO{HT}$ with $t'>t$ (in fact, $t'$ could also be much larger than $HT$) such that with constant probability, the $t$-th vertex is marked, and the $t'$-th vertex is unmarked. In this section, we reduce this problem to a simple combinatorial statement, which we prove in Lemma \ref{lem:combinatorial}.

Let $Y=(Y_i)_{i=0}^{\infty}$ be a Markov chain evolving according to $\PM$ starting from $Y_0\sim \uppi$.\footnote{Since we start in the stationary distribution actually this distribution is also translationally invariant and is the same if we look forward or backward -- due to reversibility. However our Corollary~\ref{cor:UMU} does not use these properties -- using these one might be able to prove a stronger $\Omega(1)$ lower bound for a well-chosen value of $s$.} Let $Y(s)=(Y_i(s))_{i=0}^{\infty}$ be defined to be the same chain as $Y$, except that for every marked vertex in $Y$, $Y(s)$ stays in that vertex for a length of time that is geometrically distributed with parameter $1-s$ (mean $\frac{1}{1-s}$). More precisely, let $k_1<k_2<\dots$ be the indices such that $Y_{k_j}$ is marked, and let $L_1,L_2,\dots$ be geometric random variables with mean $\frac{1}{1-s}$. Then if $\bar{L}_j = \sum_{j'=1}^j(L_{j'}-1)$,
$$Y_i(s) = \left\{\begin{array}{ll}
Y_i  & \mbox{if }i\in\{0,\dots,k_1\}\\
Y_{k_j} & \mbox{if }i\in\{k_j+\bar{L}_{j-1},\dots, k_j+\bar{L}_{j}\}\\
Y_{i-\bar{L}_j} & \mbox{if }i\in \{k_j+\bar{L}_j+1 ,\dots, k_{j+1}+\bar{L}_j\}.
\end{array}\right.$$

\begin{figure}[H]
\small
\begin{tabular}{cccccccccccccccccccccc}
$i$         	& 0 & 1 & 2 & 3 & 4 & 5 & 6 & 7 & 8 & 9 &10 &11 &12 &13 &14 &15 &16 &17 &18 &19 &20 \\
$Y_i$     	& $\sf 0$ & $\sf 1$ & $\sf 2$ & $\sf 3$ & $\color{red}\sf 4$ & $\sf 3$ & $\sf 2$ & $\sf 3$ & $\color{red}\sf 4$ & $\sf 3$ & $\color{red}\sf 4$ & $\sf 3$ & $\sf 2$  & $\sf 1$ &  \dots  & & & &     &     &     \\
$Y_i(s)$ 	& $\sf 0$ & $\sf 1$ & $\sf 2$ & $\sf 3$ & $\color{red}\sf 4$ & $\color{red}\sf 4$ & $\color{red}\sf 4$ & $\color{red}\sf 4$ & $\sf 3$ & $\sf 2$ & $\sf 3$ & $\color{red}\sf 4$ & $\color{red}\sf 4$ & $\color{red}\sf 4$   & $\sf 3$ & $\color{red}\sf 4$  & $\color{red}\sf 4$  & $\color{red}\sf 4$  & $\sf 3$ & $\sf 2$  & $\sf 1$   \\ 
\end{tabular}
\caption{Example of $Y(s)$ and $Y$ when $\PM$ is a walk on a line, $s=\frac{3}{4}$, and $\sf \color{red}4$ is marked.}
\end{figure}

It is easy to see that the marginal distribution on $Y(s)$ is a Markov chain evolving according to $\PM(s)$ starting from $\uppi$.\footnote{What we have actually described is a \emph{coupling} of the random variables $Y$ and $Y(s)$ obtained from starting in the distribution $\uppi$ and running $\PM$ and $\PM(s)$ respectively. However, it is not the same kind of coupling that is commonly used between Markov chains, as the chains start together, but do not remain together.} 
We are only interested in whether each state in the chain is marked, so we consider random variables $\bar{Y}_i,\bar{Y}_i(s)$ supported on $\{\mathrm{marked},\mathrm{unmarked}\}$. Then we are interested in lower bounding
\begin{equation}
\Pr\left(\bar{Y}_0(s)=\mathrm{unmarked}, \bar{Y}_t(s)=\mathrm{marked}, \bar{Y}_{t'}(s)=\mathrm{unmarked}\right).\label{eq:probMUbarY}
\end{equation}

A sequence of the random variables $\bar{Y}$ can be represented visually by a sequence of boxes, each of which is either unmarked (white) or marked (black). Then $\bar{Y}(s)$ is the same sequence, except that every black box is replaced with a string of black boxes, whose length is geometrically distributed with mean $r=\frac{1}{1-s}$. Thus, a good approximation of the sequence $\bar{Y}(s)$ is obtained by starting with $\bar{Y}$ and replacing each black box by a black box of length $r$, which we call an \emph{$r$-rescaling} of $\bar{Y}$, and denote $\bar{Y}^{(r)}$. Note that $r$ need not be integral, but it is convenient and sufficient to assume that it is. 

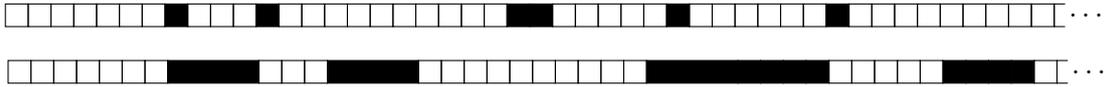
\begin{figure}[H]
\centering
\begin{tikzpicture}[scale=.3]
\draw (0,0) rectangle (1,1);
\draw (1,0) rectangle (2,1);
\draw (2,0) rectangle (3,1);
\draw (3,0) rectangle (4,1);
\draw (4,0) rectangle (5,1);
\draw (5,0) rectangle (6,1);
\draw (6,0) rectangle (7,1);
\filldraw (7,0) rectangle (8,1);
\draw (8,0) rectangle (9,1);
\draw (9,0) rectangle (10,1);
\draw (10,0) rectangle (11,1);
\filldraw (11,0) rectangle (12,1);
\draw (12,0) rectangle (13,1);
\draw (13,0) rectangle (14,1);
\draw (14,0) rectangle (15,1);
\draw (15,0) rectangle (16,1);
\draw (16,0) rectangle (17,1);
\draw (17,0) rectangle (18,1);
\draw (18,0) rectangle (19,1);
\draw (19,0) rectangle (20,1);
\draw (20,0) rectangle (21,1);
\draw (21,0) rectangle (22,1);
\filldraw (22,0) rectangle (23,1);
\filldraw (23,0) rectangle (24,1);
\draw (24,0) rectangle (25,1);
\draw (25,0) rectangle (26,1);
\draw (26,0) rectangle (27,1);
\draw (27,0) rectangle (28,1);
\draw (28,0) rectangle (29,1);
\filldraw (29,0) rectangle (30,1);
\draw (30,0) rectangle (31,1);
\draw (31,0) rectangle (32,1);
\draw (32,0) rectangle (33,1);
\draw (33,0) rectangle (34,1);
\draw (34,0) rectangle (35,1);
\draw (35,0) rectangle (36,1);
\filldraw (36,0) rectangle (37,1);
\draw (37,0) rectangle (38,1);
\draw (38,0) rectangle (39,1);
\draw (39,0) rectangle (40,1);
\draw (40,0) rectangle (41,1);
\draw (41,0) rectangle (42,1);
\draw (42,0) rectangle (43,1);
\draw (43,0) rectangle (44,1);
\draw (44,0) rectangle (45,1);
\draw (45,0) rectangle (46,1);
\draw (46,0) rectangle (47,1);
\filldraw[white] (46.5,-.1) rectangle (47.1,1.1);
\node at (47.5,.5) {$\dots$};

\node at (24.5,-2){\begin{tikzpicture}[scale=.3]
\draw (0,0) rectangle (1,1);
\draw (1,0) rectangle (2,1);
\draw (2,0) rectangle (3,1);
\draw (3,0) rectangle (4,1);
\draw (4,0) rectangle (5,1);
\draw (5,0) rectangle (6,1);
\draw (6,0) rectangle (7,1);
\filldraw (7,0) rectangle (11,1);
\draw (11,0) rectangle (12,1);
\draw (12,0) rectangle (13,1);
\draw (13,0) rectangle (14,1);
\filldraw (14,0) rectangle (18,1);
\draw (18,0) rectangle (19,1);
\draw (19,0) rectangle (20,1);
\draw (20,0) rectangle (21,1);
\draw (21,0) rectangle (22,1);
\draw (22,0) rectangle (23,1);
\draw (23,0) rectangle (24,1);
\draw (24,0) rectangle (25,1);
\draw (25,0) rectangle (26,1);
\draw (26,0) rectangle (27,1);
\draw (27,0) rectangle (28,1);
\filldraw (28,0) rectangle (32,1);
\filldraw (32,0) rectangle (33,1);
\filldraw (33,0) rectangle (34,1);
\filldraw (34,0) rectangle (35,1);
\filldraw (35,0) rectangle (36,1);
\draw (36,0) rectangle (37,1);
\draw (37,0) rectangle (38,1);
\draw (38,0) rectangle (39,1);
\draw (39,0) rectangle (40,1);
\draw (40,0) rectangle (41,1);
\filldraw (41,0) rectangle (42,1);
\filldraw (42,0) rectangle (43,1);
\filldraw (43,0) rectangle (44,1);
\filldraw (44,0) rectangle (45,1);
\draw (45,0) rectangle (46,1);
\draw (46,0) rectangle (47,1);
\filldraw[white] (46.5,-.1) rectangle (47.1,1.1);
\node at (47.5,.5) {$\dots$};
%\draw (47,0) rectangle (48,1);
%\filldraw (48,0) rectangle (52,1);
%\draw (37,0) rectangle (38,1);
%\draw (38,0) rectangle (39,1);
%\draw (39,0) rectangle (40,1);
\end{tikzpicture}};
\end{tikzpicture}
\caption{The first row shows a sequence drawn from $\bar{Y}$, and the second its $r$-rescaling, for $r=4$. The first row represents the sequence of unmarked and marked states visited by $\PM$, and the second is an approximation of the sequence of unmarked and marked states of $\PM(s)$ for $s=1-\frac{1}{r}$. 
}\label{fig:sequence}
\end{figure}

It will be sufficient to show that for some random choices $t,t'=\bigO{HT}$ with $t'>t$, we have both 
\begin{description}
\item[(1)] $\bar{Y}_t^{(r)}=\mathrm{marked}$ and 
\item[(2)] $\bar{Y}_{t'}^{(r)}=\mathrm{unmarked}$,  
\end{description}
with $\widetilde{\Omega}(1)$ probability (in $\bar{Y}$ and the randomness used to choose $t$ and $t'$) 
for some $r=\frac{1}{1-s}$.
Let $M_{\bar{Y}}^{(r)}[a,b]$ (resp.\ $U_{\bar{Y}}^{(r)}[a,b]$) be the set of $i\in \{a+1,a+2,\dots,b\}$ such that $\bar{Y}^{(r)}_i=\mathrm{marked}$ (resp.\ $\bar{Y}^{(r)}_i=\mathrm{unmarked}$). If we choose $t$ uniformly at random from some interval $\{a+1,\dots,b\}$, and $t'$ uniformly at random from some interval $\{a'+1,\dots,b'\}$, with $a'\geq b$, then it is sufficient to show that for a good choice of $r$, with constant probability in $\bar{Y}$, $|M_{\bar{Y}}^{(r)}[a,b]|/(b-a)$ and $|U_{\bar{Y}}^{(r)}[a',b']|/(b'-a')$ are $\widetilde{\Omega}(1)$.

Let $T=\lceil 3HT\rceil$, and suppose for the sake of this discussion that a marked vertex has no marked neighbour in $\PM$. This can be arranged by making two copies of the graph, ensuring that each transition switches from one copy of the graph to the other, and only considering the marked vertices in one copy to be marked (although we will ultimately not need this assumption).  In that case, for any even length interval $\{a+1,\dots,b\}$, the proportion of $t\in \{a+1,\dots,b\}$ such that $\bar{Y}_t=\mathrm{marked}$ is at most $\frac{1}{2}$. 

As a first attempt, suppose we choose $t$ uniformly at random from $\{1,\dots,2T\}$ and $t'$ uniformly at random from $\{2T+1,\dots,4T\}$. First note that, without any rescaling (i.e.\ with $r=1$), condition (2) holds, because $|M_{\bar{Y}}^{(1)}[2T,4T]|/2T\leq \frac{1}{2}$. 
It is also easy to see that upon running the non-interpolated walk, $\PM$, with high probability there will be a marked vertex in the first subsequence of length $T$. Thus, if we choose $s\geq 1-\frac{1}{T}$ so that $r\geq T$, then with high probability $|M_{\bar{Y}}^{(r)}[0,2T]|\geq T$, so condition (1) holds. However, after this rescaling, (2) might no longer hold. If, $\bar{Y}_t$ looks something like the first line of Figure \ref{fig:bad-sequence}, then before scaling (2) holds, but not (1), and after scaling by $r=T$, (1) holds, but not (2). 

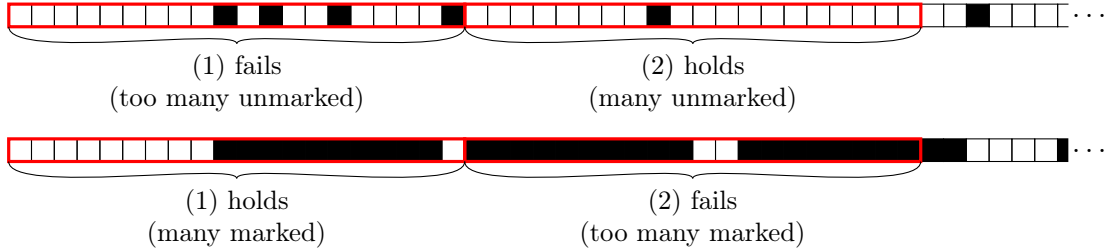
\begin{figure}[H]
\centering
\begin{tikzpicture}[scale=.3]
\draw (0,0) rectangle (1,1);
\draw (1,0) rectangle (2,1);
\draw (2,0) rectangle (3,1);
\draw (3,0) rectangle (4,1);
\draw (4,0) rectangle (5,1);
\draw (5,0) rectangle (6,1);
\draw (6,0) rectangle (7,1);
\draw (7,0) rectangle (8,1);
\draw (8,0) rectangle (9,1);
\filldraw (9,0) rectangle (10,1);
\draw (10,0) rectangle (11,1);
\filldraw (11,0) rectangle (12,1);
\draw (12,0) rectangle (13,1);
\draw (13,0) rectangle (14,1);
\filldraw (14,0) rectangle (15,1);
\draw (15,0) rectangle (16,1);
\draw (16,0) rectangle (17,1);
\draw (17,0) rectangle (18,1);
\draw (18,0) rectangle (19,1);
\filldraw (19,0) rectangle (20,1);
\draw (20,0) rectangle (21,1);
\draw (21,0) rectangle (22,1);
\draw (22,0) rectangle (23,1);
\draw (23,0) rectangle (24,1);
\draw (24,0) rectangle (25,1);
\draw (25,0) rectangle (26,1);
\draw (26,0) rectangle (27,1);
\draw (27,0) rectangle (28,1);
\filldraw (28,0) rectangle (29,1);
\draw (29,0) rectangle (30,1);
\draw (30,0) rectangle (31,1);
\draw (31,0) rectangle (32,1);
\draw (32,0) rectangle (33,1);
\draw (33,0) rectangle (34,1);
\draw (34,0) rectangle (35,1);
\draw (35,0) rectangle (36,1);
\draw (36,0) rectangle (37,1);
\draw (37,0) rectangle (38,1);
\draw (38,0) rectangle (39,1);
\draw (39,0) rectangle (40,1);
\draw (40,0) rectangle (41,1);
\draw (41,0) rectangle (42,1);
\filldraw (42,0) rectangle (43,1);
\draw (43,0) rectangle (44,1);
\draw (44,0) rectangle (45,1);
\draw (45,0) rectangle (46,1);
\draw (46,0) rectangle (47,1);
\filldraw[white] (46.5,-.1) rectangle (47.1,1.1);
\node at (47.5,.5) {$\dots$};

\draw[very thick,red] (0,0) rectangle (20,1);
\draw[very thick,red] (20,0) rectangle (40,1);

\draw (0,0) .. controls (0,-1.5) and (10,0) .. (10,-.75);
\draw (20,0) .. controls (20,-1.5) and (10,0) .. (10,-.75);
\node at (10,-2.5) {\parbox{2in}{\small \centering (1) fails\\ (too many unmarked)}};

\draw (20,0) .. controls (20,-1.5) and (30,0) .. (30,-.75);
\draw (40,0) .. controls (40,-1.5) and (30,0) .. (30,-.75);
\node at (30,-2.5) {\parbox{2in}{\small\centering (2) holds\\ (many unmarked)}};

\node at (24.35,-7.5){\begin{tikzpicture}[scale=.3]
\draw (0,0) rectangle (1,1);
\draw (1,0) rectangle (2,1);
\draw (2,0) rectangle (3,1);
\draw (3,0) rectangle (4,1);
\draw (4,0) rectangle (5,1);
\draw (5,0) rectangle (6,1);
\draw (6,0) rectangle (7,1);
\draw (7,0) rectangle (8,1);
\draw (8,0) rectangle (9,1);
\filldraw (9,0) rectangle (10,1);
\filldraw (10,0) rectangle (11,1);
\filldraw (11,0) rectangle (12,1);
\filldraw (12,0) rectangle (13,1);
\filldraw (13,0) rectangle (14,1);
\filldraw (14,0) rectangle (15,1);
\filldraw (15,0) rectangle (16,1);
\filldraw (16,0) rectangle (17,1);
\filldraw (17,0) rectangle (18,1);
\filldraw (18,0) rectangle (19,1);
\draw (19,0) rectangle (20,1);
\filldraw (20,0) rectangle (21,1);
\filldraw (21,0) rectangle (22,1);
\filldraw (22,0) rectangle (23,1);
\filldraw (23,0) rectangle (24,1);
\filldraw (24,0) rectangle (25,1);
\filldraw (25,0) rectangle (26,1);
\filldraw (26,0) rectangle (27,1);
\filldraw (27,0) rectangle (28,1);
\filldraw (28,0) rectangle (29,1);
\filldraw (29,0) rectangle (30,1);
\draw (30,0) rectangle (31,1);
\draw (31,0) rectangle (32,1);
\filldraw (32,0) rectangle (33,1);
\filldraw (33,0) rectangle (34,1);
\filldraw (34,0) rectangle (35,1);
\filldraw (35,0) rectangle (36,1);
\filldraw (36,0) rectangle (37,1);
\filldraw (37,0) rectangle (38,1);
\filldraw (38,0) rectangle (39,1);
\filldraw (39,0) rectangle (40,1);
\filldraw (40,0) rectangle (41,1);
\filldraw (41,0) rectangle (42,1);
\draw (42,0) rectangle (43,1);
\draw (43,0) rectangle (44,1);
\draw (44,0) rectangle (45,1);
\draw (45,0) rectangle (46,1);
\filldraw (46,0) rectangle (47,1);
\filldraw[white] (46.5,-.1) rectangle (47.1,1.1);
\node at (47.5,.5) {$\dots$};
%\draw (47,0) rectangle (48,1);
%\filldraw (48,0) rectangle (52,1);
%\draw (37,0) rectangle (38,1);
%\draw (38,0) rectangle (39,1);
%\draw (39,0) rectangle (40,1);

\draw[very thick,red] (0,0) rectangle (20,1);
\draw[very thick,red] (20,0) rectangle (40,1);

\draw (0,0) .. controls (0,-1.5) and (10,0) .. (10,-.75);
\draw (20,0) .. controls (20,-1.5) and (10,0) .. (10,-.75);
\node at (10,-2.5) {\parbox{2in}{\small \centering (1) holds\\ (many marked)}};

\draw (20,0) .. controls (20,-1.5) and (30,0) .. (30,-.75);
\draw (40,0) .. controls (40,-1.5) and (30,0) .. (30,-.75);
\node at (30,-2.5) {\parbox{2in}{\small\centering (2) fails\\ (too many marked)}};

\end{tikzpicture}};
\end{tikzpicture}
\caption{Illustration of the trade-off in the choice of the rescaling.}\label{fig:bad-sequence}
\end{figure}

The difficulty is that by scaling, as we create more marked boxes, we are pushing unmarked boxes out of the intervals of concern. There is a bijection between the $i^{\rm th}$ unmarked box in $\bar{Y}$ and the $i^{\rm th}$ unmarked box in $\bar{Y}^{(r)}$, but its overall position may have increased. To make this precise, let $\sigma_r(i)\in\mathbb{N}$ be the position of the $i^{\rm th}$ unmarked box in $\bar{Y}^{(r)}$. This is clearly either constant (if no marked box occurs before $\bar{Y}_i$) or strictly increasing in $r$ (otherwise). In particular, if $m(i)$ denotes the number of marked boxes before the $i^{\rm th}$ unmarked box in $\bar{Y}$, then $\sigma_r(i)=i+m(i)r$ is linear in $r$. This suggests that for small enough values $i$, as long as $m(i)\geq 1$ --- that is, there exists $j<i$ such that $\bar{Y}_j=\mathrm{marked}$ --- there should be a good choice of $r$ that pushes $\sigma_r(i)$ into the range from which we choose $t'$. 

Our second (and final) strategy will be to choose $t$ uniformly at random from $\{1,\dots,3T\}$, and $t'$ uniformly at random from $\{6T+1,\dots,12T\}$. Begin by scaling up by $r_0$, the largest scaling factor less than $3T$ such that $|M^{(r)}_{\bar{Y}}[T,3T]|/(2T)\leq \frac{3}{4}$ (for the sake of discussion, suppose it's exactly $\frac{3}{4}$). Then condition (1) holds for $r_0$, and this remains true even if we increase $r$. 

It may not be the case that scaling by $r_0$ ensures that condition (2) holds with constant probability. However, since $|U^{(r)}_{\bar{Y}}[T,3T]|/(2T)=\frac{1}{4}$, there are $\Theta(T)$ values $i$ with $\sigma_{r_0}(i)=i+m(i)r_0 \in \{T+1,\dots,3T\}$. Increasing $r$ will only increase the number of marked vertices in $\{1,\dots,3T\}$, increasing the probability of condition (1), but as marked vertices are being added to the window $\{1,\dots,3T\}$, they are pushing unmarked vertices to further positions. For a high enough value of $r$ (but not too high) we will push the $i^{\rm th}$ unmarked vertex into the window $\{6T+1,\dots,12T\}$. We can imagine searching for this good value $r$ by beginning with $r_0$ and repeatedly doubling it, as shown in Figure \ref{fig:moving-sequence}.

\begin{figure}[H]
\centering
\begin{tikzpicture}[scale=.275]

\node at (0,0){\begin{tikzpicture}[scale=.275]
\draw (0,0) rectangle (1,1);
\draw (1,0) rectangle (2,1);
\filldraw (2,0) rectangle (3,1);
\draw (3,0) rectangle (4,1);
\draw (4,0) rectangle (5,1);
\node at (4.5,.5) {\color{blue}$\star$};
\filldraw (5,0) rectangle (6,1);
\filldraw (6,0) rectangle (7,1);
\draw (7,0) rectangle (8,1);
\node at (7.5,.5) {\color{blue}$\star$};
\draw (8,0) rectangle (9,1);
\node at (8.5,.5) {\color{blue}$\star$};
\filldraw (9,0) rectangle (10,1);
\draw (10,0) rectangle (11,1);
\node at (10.5,.5) {\color{blue}$\star$};
\filldraw (11,0) rectangle (12,1);
\draw (12,0) rectangle (13,1);
\draw (13,0) rectangle (14,1);
\filldraw (14,0) rectangle (15,1);
\filldraw (15,0) rectangle (16,1);
\filldraw (16,0) rectangle (17,1);
\draw (17,0) rectangle (18,1);
\filldraw (18,0) rectangle (19,1);
\filldraw (19,0) rectangle (20,1);
\draw (20,0) rectangle (21,1);
\draw (21,0) rectangle (22,1);
\filldraw (22,0) rectangle (23,1);
\draw (23,0) rectangle (24,1);
\draw (24,0) rectangle (25,1);
\filldraw (25,0) rectangle (26,1);
\draw (26,0) rectangle (27,1);
\draw (27,0) rectangle (28,1);
\filldraw (28,0) rectangle (29,1);
\filldraw (29,0) rectangle (30,1);
\draw (30,0) rectangle (31,1);
\filldraw (31,0) rectangle (32,1);
\filldraw (32,0) rectangle (33,1);
\draw (33,0) rectangle (34,1);
\filldraw (34,0) rectangle (35,1);
\filldraw (35,0) rectangle (36,1);
\filldraw (36,0) rectangle (37,1);
\draw (37,0) rectangle (38,1);
\filldraw (38,0) rectangle (39,1);
\filldraw (39,0) rectangle (40,1);
\draw (40,0) rectangle (41,1);
\draw (41,0) rectangle (42,1);
\filldraw (42,0) rectangle (43,1);
\draw (43,0) rectangle (44,1);
\filldraw (44,0) rectangle (45,1);
\filldraw (45,0) rectangle (46,1);
\draw (46,0) rectangle (47,1);
\draw (47,0) rectangle (48,1);
%\draw (48,0) rectangle (49,1);
%\draw (49,0) rectangle (50,1);
%\draw (50,0) rectangle (51,1);
%\draw (51,0) rectangle (52,1);
\node at (53.5,.5) {\color{white}$\dots$};

\draw[very thick,red] (4,0) rectangle (12,1);
\draw[very thick,red] (24,0) rectangle (48,1);
\end{tikzpicture}};

\node at (0,-2.5){\begin{tikzpicture}[scale=.275]
\draw (0,0) rectangle (1,1);
\draw (1,0) rectangle (2,1);
\filldraw (2,0) rectangle (3,1);
\filldraw (3,0) rectangle (4,1);
\draw (4,0) rectangle (5,1);
\draw (5,0) rectangle (6,1);
\node at (5.5,.5) {\color{blue}$\star$};
\filldraw (6,0) rectangle (7,1);
\filldraw (7,0) rectangle (8,1);
\filldraw (8,0) rectangle (9,1);
\filldraw (9,0) rectangle (10,1);
\draw (10,0) rectangle (11,1);
\node at (10.5,.5) {\color{blue}$\star$};
\draw (11,0) rectangle (12,1);
\node at (11.5,.5) {\color{blue}$\star$};
\filldraw (12,0) rectangle (13,1);
\filldraw (13,0) rectangle (14,1);
\draw (14,0) rectangle (15,1);
\node at (14.5,.5) {\color{blue}$\star$};
\filldraw (15,0) rectangle (16,1);
\filldraw (16,0) rectangle (17,1);
\draw (17,0) rectangle (18,1);
\draw (18,0) rectangle (19,1);
\filldraw (19,0) rectangle (20,1);
\filldraw (20,0) rectangle (21,1);
\filldraw (21,0) rectangle (22,1);
\filldraw (22,0) rectangle (23,1);
\filldraw (23,0) rectangle (24,1);
\filldraw (24,0) rectangle (25,1);
\draw (25,0) rectangle (26,1);
\filldraw (26,0) rectangle (27,1);
\filldraw (27,0) rectangle (28,1);
\filldraw (28,0) rectangle (29,1);
\filldraw (29,0) rectangle (30,1);
\draw (30,0) rectangle (31,1);
\draw (31,0) rectangle (32,1);
\filldraw (32,0) rectangle (33,1);
\filldraw (33,0) rectangle (34,1);
\draw (34,0) rectangle (35,1);
\draw (35,0) rectangle (36,1);
\filldraw (36,0) rectangle (37,1);
\filldraw (37,0) rectangle (38,1);
\draw (38,0) rectangle (39,1);
\draw (39,0) rectangle (40,1);
\filldraw (40,0) rectangle (41,1);
\filldraw (41,0) rectangle (42,1);
\filldraw (42,0) rectangle (43,1);
\filldraw (43,0) rectangle (44,1);
\draw (44,0) rectangle (45,1);
\filldraw (45,0) rectangle (46,1);
\filldraw (46,0) rectangle (47,1);
\filldraw (47,0) rectangle (48,1);
\filldraw (48,0) rectangle (49,1);
\draw (49,0) rectangle (50,1);
\filldraw (50,0) rectangle (51,1);
\filldraw (51,0) rectangle (52,1);
\filldraw[white] (51.5,-.1) rectangle (52.1,1.1);
\node at (53,.5) {$\dots$};
%\draw (47,0) rectangle (48,1);
%\filldraw (48,0) rectangle (52,1);
%\draw (37,0) rectangle (38,1);
%\draw (38,0) rectangle (39,1);
%\draw (39,0) rectangle (40,1);

\draw[very thick,red] (4,0) rectangle (12,1);
\draw[very thick,red] (24,0) rectangle (48,1);

\end{tikzpicture}};

\node at (0,-5){\begin{tikzpicture}[scale=.275]
\draw (0,0) rectangle (1,1);
\draw (1,0) rectangle (2,1);
\filldraw (2,0) rectangle (3,1);
\filldraw (3,0) rectangle (4,1);
\filldraw (4,0) rectangle (5,1);
\filldraw (5,0) rectangle (6,1);
\draw (6,0) rectangle (7,1);
\draw (7,0) rectangle (8,1);
\node at (7.5,.5) {\color{blue}$\star$};
\filldraw (8,0) rectangle (9,1);
\filldraw (9,0) rectangle (10,1);
\filldraw (10,0) rectangle (11,1);
\filldraw (11,0) rectangle (12,1);
\filldraw (12,0) rectangle (13,1);
\filldraw (13,0) rectangle (14,1);
\filldraw (14,0) rectangle (15,1);
\filldraw (15,0) rectangle (16,1);
\draw (16,0) rectangle (17,1);
\node at (16.5,.5) {\color{blue}$\star$};
\draw (17,0) rectangle (18,1);
\node at (17.5,.5) {\color{blue}$\star$};
\filldraw (18,0) rectangle (19,1);
\filldraw (19,0) rectangle (20,1);
\filldraw (20,0) rectangle (21,1);
\filldraw (21,0) rectangle (22,1);
\draw (22,0) rectangle (23,1);
\node at (22.5,.5) {\color{blue}$\star$};
\filldraw (23,0) rectangle (24,1);
\filldraw (24,0) rectangle (25,1);
\filldraw (25,0) rectangle (26,1);
\filldraw (26,0) rectangle (27,1);
\draw (27,0) rectangle (28,1);
\draw (28,0) rectangle (29,1);
\filldraw (29,0) rectangle (30,1);
\filldraw (30,0) rectangle (31,1);
\filldraw (31,0) rectangle (32,1);
\filldraw (32,0) rectangle (33,1);
\filldraw (33,0) rectangle (34,1);
\filldraw (34,0) rectangle (35,1);
\filldraw (35,0) rectangle (36,1);
\filldraw (36,0) rectangle (37,1);
\filldraw (37,0) rectangle (38,1);
\filldraw (38,0) rectangle (39,1);
\filldraw (39,0) rectangle (40,1);
\filldraw (40,0) rectangle (41,1);
\draw (41,0) rectangle (42,1);
\filldraw (42,0) rectangle (43,1);
\filldraw (43,0) rectangle (44,1);
\filldraw (44,0) rectangle (45,1);
\filldraw (45,0) rectangle (46,1);
\filldraw (46,0) rectangle (47,1);
\filldraw (47,0) rectangle (48,1);
\filldraw (48,0) rectangle (49,1);
\filldraw (49,0) rectangle (50,1);
\draw (50,0) rectangle (51,1);
\draw (51,0) rectangle (52,1);
\filldraw[white] (51.5,-.1) rectangle (52.1,1.1);
\node at (53,.5) {$\dots$};
%\draw (47,0) rectangle (48,1);
%\filldraw (48,0) rectangle (52,1);
%\draw (37,0) rectangle (38,1);
%\draw (38,0) rectangle (39,1);
%\draw (39,0) rectangle (40,1);

\draw[very thick,red] (4,0) rectangle (12,1);
\draw[very thick,red] (24,0) rectangle (48,1);
\end{tikzpicture}};

\node at (0,-7.5){\begin{tikzpicture}[scale=.275] %4
\draw (0,0) rectangle (1,1);
\draw (1,0) rectangle (2,1);
\filldraw (2,0) rectangle (3,1);
\filldraw (3,0) rectangle (4,1);
\filldraw (4,0) rectangle (5,1);
\filldraw (5,0) rectangle (6,1);
\filldraw (6,0) rectangle (7,1);
\filldraw (7,0) rectangle (8,1);
\filldraw (8,0) rectangle (9,1);
\filldraw (9,0) rectangle (10,1);
\draw (10,0) rectangle (11,1);
\draw (11,0) rectangle (12,1);
\node at (11.5,.5) {\color{blue}$\star$};
\filldraw (12,0) rectangle (13,1);
\filldraw (13,0) rectangle (14,1);
\filldraw (14,0) rectangle (15,1);
\filldraw (15,0) rectangle (16,1);
\filldraw (16,0) rectangle (17,1);
\filldraw (17,0) rectangle (18,1);
\filldraw (18,0) rectangle (19,1);
\filldraw (19,0) rectangle (20,1);
\filldraw (20,0) rectangle (21,1);
\filldraw (21,0) rectangle (22,1);
\filldraw (22,0) rectangle (23,1);
\filldraw (23,0) rectangle (24,1);
\filldraw (24,0) rectangle (25,1);
\filldraw (25,0) rectangle (26,1);
\filldraw (26,0) rectangle (27,1);
\filldraw (27,0) rectangle (28,1);
\draw (28,0) rectangle (29,1);
\node at (28.5,.5) {\color{blue}$\star$};
\draw (29,0) rectangle (30,1);
\node at (29.5,.5) {\color{blue}$\star$};
\filldraw (30,0) rectangle (31,1);
\filldraw (31,0) rectangle (32,1);
\filldraw (32,0) rectangle (33,1);
\filldraw (33,0) rectangle (34,1);
\filldraw (34,0) rectangle (35,1);
\filldraw (35,0) rectangle (36,1);
\filldraw (36,0) rectangle (37,1);
\filldraw (37,0) rectangle (38,1);
\draw (38,0) rectangle (39,1);
\node at (38.5,.5) {\color{blue}$\star$};
\filldraw (39,0) rectangle (40,1);
\filldraw (40,0) rectangle (41,1);
\filldraw (41,0) rectangle (42,1);
\filldraw (42,0) rectangle (43,1);
\filldraw (43,0) rectangle (44,1);
\filldraw (44,0) rectangle (45,1);
\filldraw (45,0) rectangle (46,1);
\filldraw (46,0) rectangle (47,1);
\draw (47,0) rectangle (48,1);
\draw (48,0) rectangle (49,1);
\filldraw (49,0) rectangle (50,1);
\filldraw (50,0) rectangle (51,1);
\filldraw (51,0) rectangle (52,1);
\filldraw[white] (51.5,-.1) rectangle (52.1,1.1);
\node at (53,.5) {$\dots$};
%\draw (47,0) rectangle (48,1);
%\filldraw (48,0) rectangle (52,1);
%\draw (37,0) rectangle (38,1);
%\draw (38,0) rectangle (39,1);
%\draw (39,0) rectangle (40,1);

\draw[very thick,red] (4,0) rectangle (12,1);
\draw[very thick,red] (24,0) rectangle (48,1);
\end{tikzpicture}};

\node at (-28.5,0) {$r_0$};
\node at (-28.5,-2.5) {$2r_0$};
\node at (-28.5,-5) {$4r_0$};
\node at (-28.5,-7.5) {$8r_0$};
\end{tikzpicture}
\caption{As we double the scaling factor, we eventually push each unmarked vertex that began in the region $\{T+1,\dots,3T\}$, denoted by $\color{blue}\star$ symbols, into the region $\{6T+1,\dots,12T\}$, denoted by the right-most red rectangle. The same scaling doesn't work for every $\color{blue}\star$, but for every $\color{blue}\star$, there is some scaling that works.}\label{fig:moving-sequence}
\end{figure}
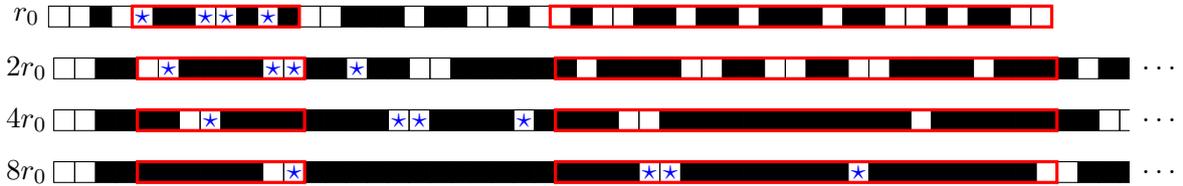

\noindent We formalize this argument with the following combinatorial lemma. 

\begin{lemma}\label{lem:combinatorial}
Let $y=(y_1,y_2,\dots)$ be a sequence of marked and unmarked boxes of length at least $12T$. Suppose that
\begin{itemize}
    \item 
    there is at least one marked among the first $T$ boxes, and
    \item
    at most $T$ of the boxes in the interval $[T, 3T]$ are marked.
\end{itemize}  If $r_0$ denotes the largest integer less than $3T$ such that $|M^{(r_0)}_y[T,3T]|\leq \frac{3}{2}T$, then letting $R=\left\{1,2,\dots,2^{\lceil \log(12T)\rceil}\right\},$
we have
$$\forall r\in R: r\geq 2r_0, |M^{(r)}_y[0,3T]|\geq \frac{3}{2}T\quad\mbox{and}\quad\sum_{r\in R: r\geq 2r_0}|U^{(r)}_y[6T,12T]|\geq \frac{1}{2}T.$$
\end{lemma}
\begin{proof}
First note that by assumption, $|M^{(1)}_y[T,3T]|\leq T$, so $r_0$ is well defined. Second, note that for $r\geq 2r_0>r_0$, since $r_0$ is maximal, $|M_y^{(r)}[0,3T]|\geq |M_y^{(r)}[T,3T]| > \frac{3}{2}T$, so the first condition holds. 

Similarly to the notation introduced before, let $y^{(r)}$ denotes the $r$-rescaling of $y$ and 
let $\sigma_r(i)$ denotes the index of the $i$-th unmarked box in $y^{(r)}$.
Then, $\sigma_r(i)=i+m(i)r$, where $m(i)$ denotes the number of marked boxes before the $i$-th unmarked box in $y$. To prove the second part of the lemma, we will show that 
\begin{equation}
\forall i : \sigma_{r_0}(i)\in \{T+1,\dots, 3T\}, \quad \exists r\in R : r\geq 2r_0\mbox{ and }\sigma_r(i)\in \{6T+1,\dots,12T\}.\label{eq:main-claim}
\end{equation}
In other words, if the $i$-th marked box in $y^{(r_0)}$ is in the interval $\{T+1,\dots,3T\}$, it gets shifted into the interval $\{6T+1,\dots,12T\}$ in $y^{(r)}$, for some $r$. Note that when $\sigma_{r_0}(i)\geq T$, we must have $m(i)\geq 1$, by the assumption that at least one of the first $T$ boxes is marked. We will show that the desired statement hold for $r=2^k$, letting $k=\lfloor \log\frac{12T-i}{m(i)}\rfloor$, so clearly $k\leq \lceil\log (12T)\rceil$. We indeed have $\sigma_r(i)=i+m(i)2^k\in \{6T+1,\dots,12T\}$:
$$i+m(i)2^k\leq i+m(i)\frac{12T-i}{m(i)} =12T
\quad\mbox{and}\quad 
i+m(i)2^k\geq i+m(i)\frac{12T-i}{2m(i)}> 6T.$$
Finally, note that since $3T\geq \sigma_{r_0}(i)=i+m(i)r_0$, we have $r_0\leq (3T-i)/m(i)$. Then
$r=2^k\geq \frac{6T-i}{m(i)}\geq 2r_0$, concluding the proof of \eqref{eq:main-claim}.

The second claim in the lemma follows from \eqref{eq:main-claim}, because 
$$|\{i:\sigma_{r_0}(i)\in \{T+1,\dots,3T\}\}|=|U^{(r_0)}_y[T,3T]|\geq \frac{1}{2}T$$ 
by definition of $r_0$. By \eqref{eq:main-claim}, each of these $\geq \frac{1}{2}T$ unmarked vertices contributes to at least one term of $\sum_{r\in R:r\geq 2r_0}|U^{(r)}_y[6T,12T]|\geq \frac{1}{2}T$. 
\end{proof}

Note that even if we replace the fixed rescalings of each marked interval with independent geometric random variables, any fixed set of marked intervals gets a total rescaling that is within a factor 2 of its expected length with probability $\frac{7}{16}$, as per the following lemma, proven in Appendix~\ref{app:geomSum}:

\begin{restatable}{lemma}{geomSum}\label{lemma:geomSum}
	Let $p\in(0,1]$, $t\in \N$ and $Z=\sum_{i=1}^{t}G_i$, where $G_i$ is a geometric random variable having parameter $p$. Then
	$$ \Pr\left(\frac{t}{2p}\leq Z \leq \frac{2t}{p} \right)\geq \frac{7}{16}.$$ 
\end{restatable}

We can now conclude with a statement about the random walk $\PM(s)$ that we will use to analyze our quantum algorithm. The final statement we need is proven in Corollary~\ref{cor:AMRMS}. We first prove the following.

\begin{corollary}\label{cor:UMU}
Let $\cal P$ be any Markov chain (not necessarily reversible).
Let $\rho$ be any distribution (not necessarily stationary). Let $E$ be the event that: the first vertex sampled according to $\rho$ is unmarked; a marked vertex is encountered within the first $T$ steps of $\cal P$ (equivalently $\PM(s)$); and at most $T$ of the next $2T$ steps of $\cal P$ (equivalently, the next $2T$ steps of $\PM(s)$ that do not consist of staying at a marked vertex) go to a marked vertex. 

Let $r\in R$, $t\in \{1,\dots,3T\}$ and $t'\in \{3T+1,\dots,24T\}$ be chosen uniformly at random, and let $s=1-\frac{1}{r}$. Then 
$$\mathbb{E}_{t,t',s}\left(\Pr_{Y_0(s)\sim \rho}(Y_0(s)\in U, Y_t(s)\in M, Y_{t'}(s)\in U|E)\right) = \Omega\left(\frac{1}{\log(T)}\right).$$
\end{corollary}
\begin{proof}
Let $S=\{1-\frac{1}{r}:r\in R\}$. 
When sampling $Y(s)$, we want to make a distinction between: 
\begin{enumerate}
    \item[(1)]
    the randomness used, when at a marked vertex, to decide whether to stay or take a step of the walk  according to $\PM$, and
    \item[(2)]
    the randomness used to decide which neighbouring vertex to transition to (assuming a step is to be taken), according to $\PM$.
\end{enumerate}
 The second type of randomness, (2), is exactly the randomness of $Y$ (recall that $Y$ is a Markov chain that is coupled to $Y(s)$ in the sense that if $Y(s)$ does not stay at the current vertex, then it moves as $Y$). Thus, we can write\footnote{We can keep all sums finite by only considering a chain $Y$ of finite length at least $24T$.}:
\begin{align}
&\mathbb{E}_{t,t',s}\left(\Pr(Y_t(s)\in M, Y_{t'}(s)\in U|E)\right)=\label{eq:exp1}\\
&= \sum_{y}\Pr(Y=y|E)\mathbb{E}_s\sum_{y(s)}\Pr(Y(s)=y(s)|Y=y)\sum_{t=1}^{3T}\frac{1}{3T}\sum_{t'=3T+1}^{24T}\frac{\Pr(y_t(s)\in M,y_{t'}(s)\in U)}{21T},\nonumber
\end{align}
noting that
\begin{align}
&\sum_{t=1}^{3T}\frac{1}{3T}\sum_{t'=3T+1}^{24T}\frac{1}{21T}\Pr(y_t(s)\in M,y_{t'}(s)\in U)=\label{eq:exp2}\\
&=\frac{|\{t\in\{1,\dots,3T\}:y_t(s)\in M\}|}{3T}\frac{|\{t'\in\{3T+1,\dots,24T\}:y_{t'}(s)\in U\}|}{21T}.\nonumber
\end{align}

For a fixed path of Markov chain $Y(s)$, $y(s)$, suppose the average over marked vertices encountered in the first $3T$ steps, number of steps spent at the marked vertex is at least $r/2$, for $r=\frac{1}{1-s}$. Then we have:
$$|\{t\in \{1,\dots,3T\}:Y_t(s)\in M\}|\geq \frac{1}{2}|{M_{y}^{(r)}[0,3T]}|.$$
To see this, note that increasing one of the marked regions that begins in $[0,3T]$ by 1, we cannot decrease the number of marked vertices in $[0,3T]$, and decreasing one of the marked regions by 1, can only decrease the number of marked boxes in $[0,3T]$ by 1.

Moreover, suppose the average over marked vertices encountered by $y(s)$ in $[0,6T]$, number of steps spent at the marked vertex, is at least $r/2$ and the average in $[0,12T]$ is at most $2r$. Then the unmarked vertices of $U_{y}^{(r)}[6T,12T]$ may be moved and spread out, but they will all occur within the range $\{3T+1,\dots,24T\}$. Thus:
$$|\{t'\in \{3T+1,\dots,24T\}:Y_{t'}(s)\in U\}|\geq |U_{y}^{(r)}[6T,12T]|.$$
Let $F$ be the event that all of these conditions hold, that is, the average length of stay at a marked vertex in steps $\{1,\dots,3T\}$ and $\{1,\dots,6T\}$ is at least $r/2$, and the average length of stay at a marked vertex in steps $\{1,\dots,12T\}$ is at most $2r$. Then by Lemma~\ref{lemma:geomSum}, $\Pr(F)\geq (7/16)^3$. Thus, continuing from \eqref{eq:exp1} and \eqref{eq:exp2}, we have:
\begin{align*}
&\mathbb{E}_{t,t',s}\left(\Pr_{Y_0(s)\sim\rho}(Y_t(s)\in M, Y_{t'}(s)\in U|E)\right)\\
&\geq\Omega(1)\sum_{y}\Pr(Y=y|E)\sum_{s\in S}\frac{1}{|S|}\sum_{y(s)}\Pr(Y(s)=y(s)|Y=y,F)\frac{|M_y^{(\frac{1}{1-s})}[0,3T]|}{6T}\frac{|U_y^{(\frac{1}{1-s})}[6T,12T]|}{21T}\\
&= \Omega(1)\frac{1}{|S|}\sum_{y}\Pr(Y=y|E)\sum_{s\in S}\frac{|M_y^{(\frac{1}{1-s})}[0,3T]|}{6T}\frac{|U_y^{(\frac{1}{1-s})}[6T,12T]|}{21T}\\
&\geq\Omega(1)\frac{1}{|S|}\sum_y\Pr(Y=y|E)\frac{1}{4}\frac{1}{42}, \quad\mbox{ by Lemma \ref{lem:combinatorial}}\\
&=\Omega\left(\frac{1}{|S|}\right)=\Omega\left(\frac{1}{\log T}\right). \qedhere
\end{align*}
\end{proof}

We can now conclude with the statement we will need in the analysis of our algorithm in Section~\ref{subsec:final-alg}.

\begin{corollary}\label{cor:AMRMS}
	Let $\cal P$ be a reversible ergodic Markov chain, and let $\uppi$ be its stationary distribution.
	If $p_M\leq 1/9$ and $T\geq 3HT$, then choosing $s\in S=\{1-\frac{1}{r}:r\in R\}$ and $t\in [24T]$ uniformly at random we get, that 
	$$\mathbb{E}\left[\nrm{\Pi_M D^t(s)\ket{\sqrt{\uppi_U}}}^2\right]= \Omega\left(\frac{1}{\log(T)}\right).$$
	\gnote{It could be improved to $\Omega\left(\frac{1}{\log(T\cdot p_M)}\right)$, since it is enough to consider $R\cap [1/p_M, T]$.}
\end{corollary}
\begin{proof}
	First we prove that the event $E$ in Corollary~\ref{cor:UMU} holds with constant probability. The probability that the initial vertex is marked is $p_M\leq 1/9$. The probability that the Markov chain does not hit a marked vertex in $T\geq 3 HT$ steps is at most $1/3$ by Markov's inequality. Finally, the expected number of marked sites in the first $3T$ steps is $p_M 3T\leq T/3$, therefore the probability that there are more than $T$ marked vertices in the first $3T$ steps is at most $1/3$ by Markov's inequality. By the union bound we get the probability of the complement of $E$ is at most $1/9+1/3+1/3=7/9$, therefore $E$ holds with probability at least $2/9$.
	
	Let us define $\ket{v^t(s)}:=\Pi_M D^t(s)\ket{\sqrt{\uppi_U}}$, then by Corollary~\ref{cor:UMU}, recalling that $\hat{t}=t'-t$, we have that
	\begin{align*}
	\Omega(1)&=\sum_{s\in S}\sum_{t, \hat{t} \in [24T]}\sum_{x,z\in U}\frac{\uppi_x\bra{x}\PM^t(s)\Pi_M\PM^{\hat{t}}(s)\ket{z} \tag*{by Corollary~\ref{cor:UMU}}}{(24T)^2}\\	
	&=\sum_{s\in S}\sum_{t, \hat{t} \in [24T]}\frac{\bra{\sqrt{\uppi_U}}D^t(s)\Pi_M D^{\hat{t}}(s)\ket{\sqrt{\uppi_U}}}{(24T)^2} \tag*{by Eq.~\eqref{eq:probMUPrime}}\\	
	&=\sum_{s\in S}\sum_{t, \hat{t} \in [24T]}\frac{\braket{v^t(s)}{v^{\hat{t}}(s)}}{(24T)^2}
	\leq\sum_{s\in S}\sum_{t, \hat{t} \in [24T]}\frac{\nrm{v^t(s)}\nrm{v^{\hat{t}}(s)}}{(24T)^2} \tag*{by Cauchy-Schwartz}\\
	&=\sum_{s\in S}\left(\sum_{t \in [24T]}\frac{\nrm{v^t(s)}}{24T}\right)^{\!\!2}	
	\leq\sum_{s\in S}\sum_{t \in [24T]}\frac{\nrm{v^t(s)}^2}{24T},	
	%\tag*{since arithmetic mean $\leq$ root-mean square}
	\end{align*}
	where the last inequality follows from the fact that the arithmetic mean is always majorated by the root-mean square.
\end{proof}

\subsection{The final algorithm and its analysis}\label{subsec:final-alg}

We can now present our fast-forwarding-based algorithm, proving Theorem~\ref{thm:main}. Recall that $S=\{1-\frac{1}{r}:r\in R\}$, where $R=\{1,2,\dots,2^{\lceil \log(12T)\rceil}\}$.
The full algorithm is as follows. 

\begin{algorithm}[H]
	\textbf{Search}($ \PM $, $ M $, $T$)\\ %: aim
	Use $\bigO{\!\!\sqrt{\log(T)}}\!$ rounds amplitude amplification to amplify the success probability of steps~1-$3$:
	\begin{enumerate}
		\item Use $\setup(\PM)$ to prepare the state $$\sum_{t=1}^{T}\frac{1}{\sqrt{T}}\ket{t}\sum_{s\in S}\frac{1}{\sqrt{|S|}}\ket{s}\ket{\sqrt{\uppi}}.$$
		\item Measure $\{\Pi_M,I-\Pi_M\}$ on the last register. If the outcome is ``marked'', measure in the computational basis, and output the entry in the last register. Otherwise continue with the (subnormalized) post-measurement state %\frac{1}{\sqrt{1-p_M}}
		$$\sum_{t=1}^{T}\frac{1}{\sqrt{T}}\ket{t}\sum_{s\in S}\frac{1}{\sqrt{|S|}}\ket{s}\ket{\sqrt{\uppi_U}}.$$
		\item Use quantum fast-forwarding, controlled on the first two registers, to map $\ket{t}\ket{s}\ket{\sqrt{\uppi_U}}$ to $\ket{1}\ket{t}\ket{s}D^t(s)\ket{\uppi_U}+\ket{0}\ket{\Gamma}$ for some arbitrary $\ket{\Gamma}$, with precision $\bigO{\frac{1}{\log(T)}}$. Finally, measure the last register and output its content if marked, otherwise output \texttt{No marked vertex}.
	\end{enumerate}
	\caption{Fast-forwarding-based search algorithm}\label{alg:alg2}
\end{algorithm}

If $T\geq 72 HT(\PM,M)$, then the success probability of the above steps 1-3  is $\Omega\left(\frac{1}{\log(T)}\right)$, as shown by Corollary~\ref{cor:AMRMS}. Thus, after $\bigO{\sqrt{\log(T)}}$ steps of amplitude amplification, the success probability becomes $\Omega(1)$. 

By Theorem~\ref{thm:fast-forwarding} the complexity of step 3 is $\bigO{\sqrt{T\log\log(T)}({\sf U}+{\sf C})}$, since $W(s)$ can be implemented in cost $\bigO{{\sf U}+{\sf C}}$. Thus, the complexity of steps 1-3 is
$\bigO{\mathsf{S}+\sqrt{T\log\log(T)}(\mathsf{U}+\mathsf{C})}$, where $\mathsf{S}$ is the complexity of generating $\ket{\sqrt{\uppi}}$, using $\mathtt{Setup}(\PM)$.
Amplitude amplification gives a $\sqrt{\log(T)}$ multiplicative overhead, proving Theorem~\ref{thm:main}.

\section*{Acknowledgments}
	The authors thank Simon Apers, Jérémie Roland, Shantanav Chakraborty, and Ashwin Nayak for useful comments and discussions.
	
\bibliographystyle{alphaUrlePrint}
\bibliography{Bibliography}

\appendix

\section{Proof of Lemma \ref{lem:htpLB}}\label{app:A}

\htpLB*
\begin{proof}
While the   vertices  $ (x_1,x_2) $ of the torus graph can be ordered arbitrarily, we use the lexicographic ordering (i.e., $ (x_1,x_2) \prec (x_1',x_2') $ iff $ x_1 < x_1' $ or $ x_1=x_1' $ and $x_2<x_2'  $),
Then $ \PM $  is formed accordingly to this ordering, i.e., the first row (column) of $ \PM  $ corresponds to the vertex $ (0,0) $, the second row (column) corresponds to the vertex $ (0,1) $, and so on.  Now  $ \PM $  is an $ (N^2) \times (N^2) $ BCCB  (block circulant with circulant blocks) matrix  \cite[Definition 5.27]{vogel2002CompMethInvProbs} and   can be diagonalized as \cite[Proposition 5.31]{vogel2002CompMethInvProbs}
\[
\PM  =
\lr{F_N \otimes F_N} \diag \Lambda \lr{F_N \otimes F_N}^{\dagger} ,
\]
where $ \Lambda $ is the vector of the eigenvalues of $ \PM $, $ \otimes $ stands for the Kronecker product and
\[
F_N =
\frac{1}{\sqrt N}
\begin{pmatrix}
1 & 1            & 1               & \ldots & 1                   \\
1 & \omega       & \omega^2        & \ldots & \omega^{N-1}        \\
1 & \omega^2     & \omega^4        & \ldots & \omega^{2(N-1)}     \\
&  & \ddots \\
1 & \omega^{N-1} & \omega^{2(N-1)} & \ldots & \omega^{(N-1)(N-1)}
\end{pmatrix}
, \quad \omega := \exp\lr{\frac{2\pi \, i }{N}  }.
\]
It can be verified by direct calculation (or by applying the two-dimensional discrete Fourier transform as described in  \cite[Proposition 5.31]{vogel2002CompMethInvProbs}) that the eigenvalues of the matrix $ \PM $ are
\[
\lambda_{j,k} =
\frac{1}{5} \lr{1 + 2\cos \frac{2\pi j}{N}+ 2\cos \frac{2\pi k}{N}}, \quad j,k \in \lrb{0,1,\ldots,N-1},
\]	and  the corresponding eigenvectors are
\[
\ket{ v_{j,k} } = w^{(j)} \otimes w^{(k)},
\quad
\text{where}
\quad
w^{(j)} :=
\frac{1}{\sqrt N}
\begin{pmatrix}
1 & \omega^j & \omega^{2j} & \ldots & \omega^{(N-1)j}
\end{pmatrix}^T .
\]
By \cite[Theorem 17]{krovi2010QWalkFindMarkedAnyGraph}, the  extended hitting time is related to the interpolated hitting time $ HT(0) $ via 
$ HT^+    ={p_M^{-2}} \, HT(0) $, where $ HT(0) $ is defined as 
\[ 
HT(0) =
\sum_{\substack{
		j = 0 .. N-1 \\
		k = 0 .. N-1 \\
		(j,k) \neq (0,0)
}}  \frac{\lrv{ \braket{v_{j,k}}{U}  }^2}{1 - \lambda_{j,k}},
\]
and (since the stationary distribution $ \uppi $ is uniform) $ p_M =  m/{N^2} $ and $ \ket U = \sqrt {u^{-1}} \sum_{x \in U} \ket x $; thus 
\begin{equation}\label{eq:L101}
HT^+ =  \frac{5}{4}\, \frac{N^2}{m^2 \, u} \sum_{\substack{
		j = 0 .. N-1 \\
		k = 0 .. N-1 \\
		(j,k) \neq (0,0)
}}
\frac{\bigg|\vcenter{\hbox{$\sum^{\rotatebox[origin=c]{90}{\kern1.6mm}}\limits_{(x_1,x_2) \in U}  \omega^{jx_1 + kx_2}$}}\bigg|^2
}{
	\sin^2  \frac{ \pi j}{N} + \sin^2  \frac{\pi k}{N}
}.
\end{equation}
Here we have also applied 
\[
1 - \lambda_{j,k} = \frac{1}{5} \lr{2 - 2\cos \frac{2\pi j}{N}+ 2 - 2\cos \frac{2\pi k}{N}}
=
\frac{4}{5} \lr{\sin^2  \frac{ \pi j}{N} + \sin^2  \frac{\pi k}{N}  }
\]
and
\[
\sum\limits_{x \in U}  \braket x {v_{j,k}} = \frac{1}{N} \sum_{(x_1,x_2) \in U} \omega^{jx_1 + kx_2}.
\]
For all pairs $ (j,k)  $ with $ (j,k) \neq (0,0) $, $ 0 \leq j,k \leq N-1 $, we have
\[
\sum_{x_1=0}^{N-1} \sum_{x_2=0}^{N-1}\omega^{j x_1 + k x_2}=
\sum\limits_{(x_1,x_2) \in M} \omega^{j x_1 + k x_2}  + \sum\limits_{(x_1,x_2) \in U} \omega^{j x_1 + k x_2}  = 0,
\]
hence we can rewrite \eqref{eq:L101} as
\begin{equation}\label{eq:L102}
HT^+ =  \frac{5}{4}\, \frac{N^2}{m^2 \, u} \sum_{\substack{
		j = 0 .. N-1 \\
		k = 0 .. N-1 \\
		(j,k) \neq (0,0)
}}
\frac{
	\bigg|
		\vcenter{\hbox{$\sum^{\rotatebox[origin=c]{90}{\kern1.6mm}}\limits_{(x_1,x_2) \in M}  \omega^{jx_1 + kx_2}$}}
	\bigg|^2
}{
	\sin^2  \frac{ \pi j}{N} + \sin^2  \frac{\pi k}{N}
}.
\end{equation}
Since all summands on the RHS of \eqref{eq:L102} are nonnegative, the desired bound \eqref{eq:L100} follows.
\end{proof}

\section{Concentration of sums of geometric random variables}\label{app:geomSum}

\geomSum*
\begin{proof}
Let $G$ be a geometric random variable of parameter $p$, then it has expectation value $1/p$ and variance $(1-p)/p^2\leq 1/p^2$. Moreover $\Pr(G\leq  k)=1-(1-p)^k$ for all $k\in \N$. In particular the probability that $\Pr(\lfloor 1/(2p)\rfloor< G\leq  \lfloor 2/p\rfloor)=(1-p)^{\lfloor 1/(2p)\rfloor}-(1-p)^{\lfloor 2/p\rfloor}\geq 7/16$.
%Checked by Mathematica:
%f[p_] := (1 - p)^(Floor[1/(2 p)]) - (1 - p)^(Floor[2/p]); f[1/2]
%Plot[f[p], {p, 0, 1}]
More generally $Z$ has negative binomial distribution.
One can check that for every $t\in [7]$ and all $p\in(0,1]$ we have that $\Pr(\lfloor t/(2p)\rfloor< Z\leq  \lfloor 2t/p\rfloor)\geq 7/16$, see Figure~\ref{fig:ProbPlots}.
%Checked by Mathematica:
%g[p_, r_] := CDF[NegativeBinomialDistribution[r, p], Floor[2 r/p] - r] - CDF[NegativeBinomialDistribution[r, p], Floor[r/(2 p)] - r]
%Plot[Table[g[p, r], {r, 1, 7}] // Evaluate, {p, 0, 1}, PlotRange -> {13/32, 1}, ImageSize -> {800, 800/GoldenRatio}, PlotPoints -> 325, Ticks -> {Automatic, {7/16, 10/16,13/16, 1}}, PlotLegends -> Placed[LineLegend[Table["t=" <> IntegerString[k], {k, 1, 7}]], {0.95, 0.25}]]

On the other hand the variance of $Z$ is at most $\frac{t}{p^2}$, so 
$ \Pr\left(\left|Z-t/p\right| \geq \frac{t}{2p} \right)\leq \frac{4}{t}$ by Chebyshev's inequality, which implies the claim for $t\geq 8$.
\end{proof}

\begin{figure}[ht]
	\hskip	0.00\linewidth\includegraphics[draft=false,width=0.93\linewidth]{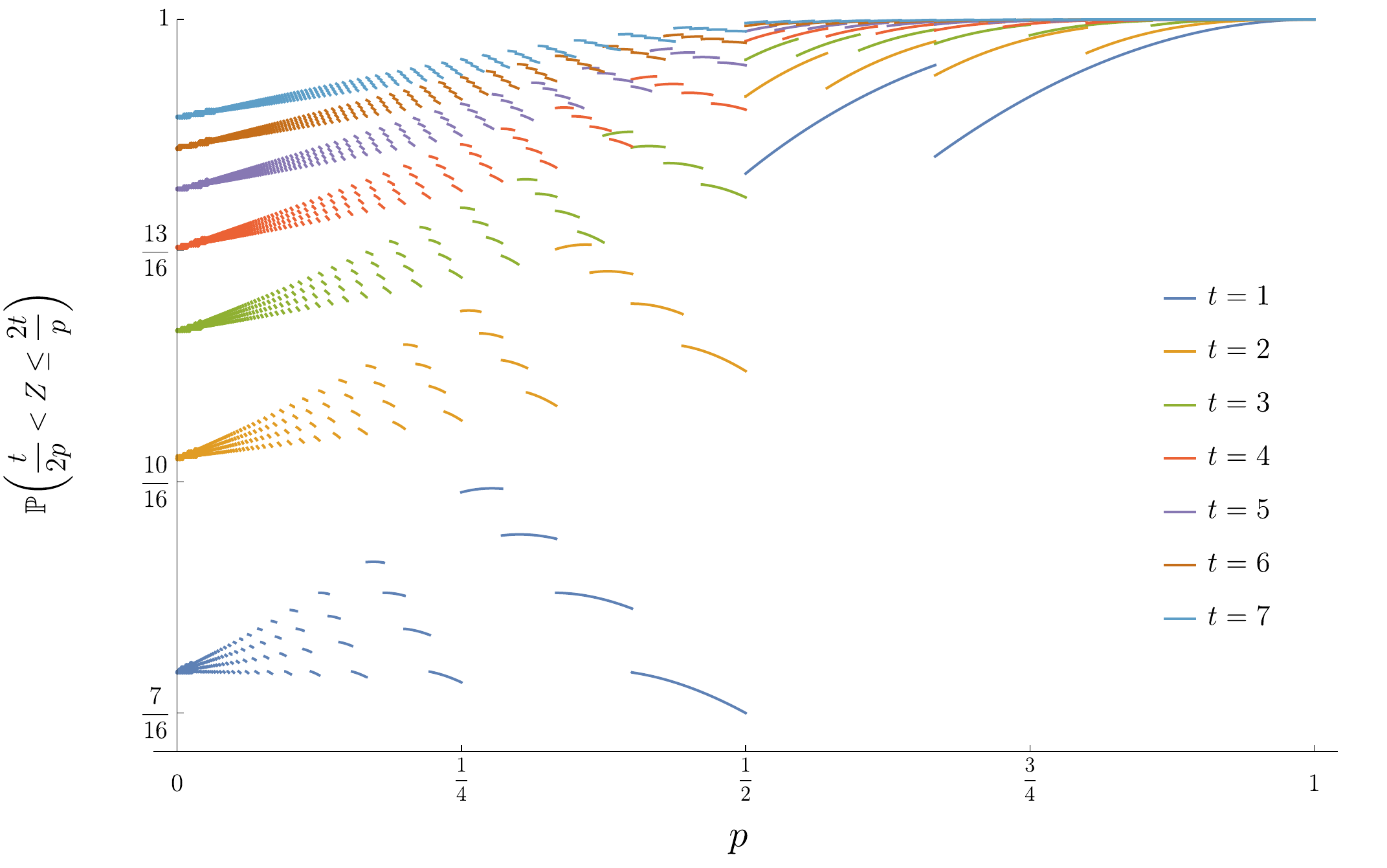}
	\caption{Illustration of Lemma~\ref{lemma:geomSum} for $t=1,2,\ldots, 7$}\label{fig:ProbPlots}
\end{figure}

\end{document}
	
	\providecommand\mywordcount{
		%Need to run with --shell-escape (or --enable-write18 for MiKTeX) flag. (E.g. pdflatex --shell-escape)
		\ifcount
		$ \phantom{\sum}$ \\ \noindent\textbf{\large Wordcount} \\ $\phantom{\sum}$ \\
		\noindent\input|"if [ -f /home/gilyen/texcount.pl ]; then /home/gilyen/texcount.pl -sub=section \jobname.tex | grep -e Words -e Number -e Section -e top -e Part | awk 1 ORS='\string\\\string\\' | sed -e 's/\string\_/ /g'; else texcount -sub=section \jobname.tex | grep -e Words -e Number -e Section -e top -e Part | awk 1 ORS='\string\\\string\\' | sed -e 's/\string\_/ /g'; fi"
		text+headers+captions (\#headers/\#floats/\#inlines/\#displayed)\\
		\else
		%Warning: Counter temporarily disabled to speed up compiling!
		\fi
	}
	
	%\mywordcount	
	
\end{document}